\documentclass[conference]{IEEEtran}
\usepackage{balance}  %
\usepackage{graphicx}
\usepackage{mathrsfs}
\usepackage{citesort,setspace}
\usepackage[usenames,dvipsnames,svgnames,table]{xcolor}

\usepackage{amsmath,amssymb,xspace,subfigure,amsthm}
\usepackage{url}
\usepackage{multirow}
\usepackage[nosort]{cite}
\usepackage{booktabs}
\usepackage{algorithm}
\usepackage{algorithmicx}
\usepackage[noend]{algpseudocode}



\newtheoremstyle{hypstyle}%
 {\topsep}%
 {\topsep}%
 {\itshape}%
 {}%
 {\bfseries}%
 {:}%
 {.5em}%
{}

\newcommand{\abs}[1]{\left| #1 \right|}
\newcommand{\set}[1]{\left\{ #1 \right\}}
\newcommand{\ignore}[1]{}

\newcommand{\myparagraphF}[1]{\noindent \underline{\bf #1:}}
\newcommand{\myparagraph}[1]{\vspace{0.5em} \noindent \underline{\bf #1:}}

\newcommand{\argmax}{\operatornamewithlimits{argmax}}

\theoremstyle{hypstyle}
\newtheorem{theorem}{Theorem}
\newtheorem{lemma}[theorem]{Lemma}

\newtheorem{corollary}[theorem]{Corollary}

\newtheorem{problem}{Problem}

\newcommand{\neighbor}[2]{N_{#2}(#1)}
\newcommand{\sketch}[2]{\widetilde{N}_{#2}(#1)}
\newcommand{\rank}[1]{r_{#1}}
\newcommand{\kth}{k\textsf{-th}}

\newcommand{\secit}{}

\begin{document}

\IEEEoverridecommandlockouts
\title{Fractality of Massive Graphs: Scalable Analysis \\ with Sketch-Based Box-Covering Algorithm$^\ast$
  \thanks{$^\ast~$This work is done while all authors were at National Institute of Informatics.
A shorter version of this paper appeared in the proceedings of ICDM 2016~\cite{short}.}}

\author{
{Takuya Akiba{\small$^{\dagger}$},
 Kenko Nakamura{\small $^{\ddag}$},
 Taro Takaguchi{\small $^{\S}$}
} \\

\fontsize{10}{10}\selectfont\itshape
$^{\dagger}$\,Preferred Networks, Inc. \;
$^{\ddag}$\,Recruit Communications Co., Ltd. \\
$^{\S}$\,National Institute of Information and Communications Technology \\

\fontsize{9}{9}\selectfont\ttfamily\upshape
akiba@preferred.jp,
kenkoooo@r.recruit.co.jp,
taro.takaguchi@nict.go.jp
}

\maketitle

\begin{abstract}
Analysis and modeling of networked objects are fundamental pieces of modern data mining.
Most real-world networks, from biological to social ones, are known to have common structural properties.
These properties allow us to model the growth processes of networks and to develop useful algorithms.
One remarkable example is the fractality of networks, which suggests the self-similar organization of global network structure.
To determine the fractality of a network, we need to solve the so-called box-covering problem, where preceding algorithms are not feasible for large-scale networks.
The lack of an efficient algorithm prevents us from investigating the fractal nature of large-scale networks.
To overcome this issue, we propose a new box-covering algorithm based on
recently emerging sketching techniques.
We theoretically show that it works in near-linear time with
a guarantee of solution accuracy.
In experiments,
we have confirmed that
the algorithm enables us to study the fractality of million-scale networks for the first time.
We have observed that its outputs are sufficiently accurate
and that its time and space requirements are orders of magnitude smaller than those of previous algorithms.
\end{abstract}

\section{Introduction}
Graph representation of real-world systems, such as social relationship, biological reactions, and hyperlink structure, gives us a strong tool to analyze and control these complex objects~\cite{Newman2010}.
For the last two decades, we have witnessed the spark of network science that unveils common structural properties across a variety of real networks.
We can exploit these frequently observed properties to model the generation processes of real networked systems~\cite{Leskovec2005} and to develop graph algorithms that are applicable to various objects~\cite{Kleinberg2000}. %
A notable example of such properties is the scale-free property~\cite{Barabasi1999,Caldarelli2007}, which manifests a power-law scaling in the vertex degree distribution and existence of well-connected vertices (often called hubs).
The scale-free property, existence of hubs especially, underlies efficient performance of practical graph algorithms on realistic networks~\cite{Akiba2013,Kempe2003}.%

Although the scale-free property inspires us to design better network models and algorithms, it is purely based on the local property of networks, i.e., the vertex degree.
Real-world networks should possess other common properties beyond the local level.
As a remarkable example of such non-local properties, the fractality of complex networks was found in network science~\cite{Song2005,Gallos2007}.
The fractality of a network suggests that the network shows a self-similar structure; if we replace groups of adjacent vertices with supervertices, the resultant network holds a similar structure to the original network (see Section~\ref{sec:problem} for its formal definition).
The fractality of networks gives us unique insights into modeling of growth processes of real-world networks~\cite{Song2006}.
In addition, fractal and non-fractal networks, even with the same degree distribution, indicate striking differences in facility of spreading~\cite{Serrano2011} and vulnerability against failure~\cite{Hasegawa2013}.
Aside from theoretical studies, the fractality provides us with useful information about network topology. Examples include the backbone structure of networks~\cite{Goh2006} and the hierarchical organization of functional modules in the Internet~\cite{Carmi2007}, metabolic~\cite{Song2006} and brain~\cite{Gallos2012} networks, to name a few.

Determination of the fractality of a network is based on the so-called box-covering problem~\cite{Song2005} (also see Section~\ref{sec:problem}).
We locally cover a group of adjacent vertices with a box such that all vertices in a box are within a given distance from each other, and then we count the number of boxes we use to cover the whole network.
In principle, we have to minimize the number of boxes that cover the network, which is known to be an NP-hard problem~(see \cite{Song2007} and references therein).
Although different heuristic algorithms are proposed in the previous work~(e.g., \cite{Song2007,Schneider2012}), they are still not so efficient as to be able to process networks with millions of vertices.
This limitation leaves the fractal nature of large-scale networks far from our understanding.

\myparagraph{Contributions}
The main contribution of the present study is to propose a new type of box-covering algorithm
that is much more scalable than previous algorithms.
In general, previous algorithms first explicitly instantiate
all boxes
and then reduce the box cover problem to the famous set cover problem.
This approach requires quadratic $\Theta(n^2)$ space
for representing neighbor sets and is obviously infeasible for large-scale networks with millions of vertices.
In contrast, the central idea underlying the proposed method
is to solve the problem in the \emph{sketch space}.
That is, we do not explicitly instantiate neighbor sets;
instead, we construct and use the \emph{bottom-$k$ min-hash sketch representation}~\cite{ads/original,ads/cohen15}
of boxes.

Technically, we introduce several new concepts and algorithms.
First, to make the sketch-based approach feasible,
we introduce a slightly relaxed problem called $(1-\epsilon)$-\textsc{BoxCover}.
We also define a key subproblem called the $(1-\epsilon)$-\textsc{SetCover} problem.
The proposed box-cover algorithm consists of two parts.
First, we generate min-hash sketches of all boxes
to reduce the $(1-\epsilon)$-\textsc{BoxCover} problem to the $(1-\epsilon)$-\textsc{SetCover} problem.
Our sketch generation algorithm does not require
explicit instantiation of actual boxes and is efficient in terms of both time and space.
Second, we apply our efficient sketch-space set-cover algorithm to obtain the final result.
Our sketch-space set-cover algorithm is based on a greedy approach,
but is carefully designed with event-driven data structure operations
to achieve near-linear time complexity.

We theoretically guarantee both the scalability and the solution quality of the proposed box-cover algorithm.
Specifically, for a given trade-off parameter $k$ and radius parameter $\ell$,
it works in $O((n + m) k \log k \min\set{\ell, \log n})$ time
and $O(nk + m)$ space.
The produced result
is a solution of \textsc{$(1-\epsilon)$-BoxCover}
within a factor $1 + 2\ln n$ of the optimum for \textsc{BoxCover}
for $\epsilon \geq 2 \sqrt{5 (\ln n) / k}$,
with a high probability that asymptotically approaches 1.

In experiments, we have confirmed the practicability of the proposed method.
First, we observed that
its outputs are quite close to those of previous algorithms
and are sufficiently accurate to
recognize networks with ground-truth fractality.
Second, the time and space requirements
are orders of magnitude smaller than those of previous algorithms,
resulting in the capability of handling large-scale networks with tens of millions of vertices and edges.
Finally, we applied our algorithm to a real-world million-scale network
and accomplish its fractality analysis for the first time.

\myparagraph{Organization}
The remainder of this paper is organized as follows.
We describe the definitions and notations in Section~\ref{sec:preliminaries}.
In Section~\ref{sec:set_cover},
we present our algorithm for sketch-space \textsc{SetCover}.
We explain our sketch construction algorithm
to complete the proposed method for \textsc{BoxCover} in Section~\ref{sec:box_cover}.
In Section~\ref{sec:improve},
we present a few empirical techniques to further improve the proposed method.
We explain the experimental evaluation of the proposed method in Section~\ref{sec:experiments}.
We conclude in Section~\ref{sec:conclusions}.

\section{Preliminaries}
\label{sec:preliminaries}

\subsection{Notations}
We focus on networks that are modeled as undirected unweighted graphs.
Let $G = (V, E)$ be a graph, where $V$ and $E$ are the vertex set and edge set, respectively.
We use $n$ and $m$ to denote $\abs{V}$ and $\abs{E}$, respectively.
For $d \geq 0$ and $v \in V$,
we define $\neighbor{v}{d}$ as the set of vertices with distance at most $d$ from $v$.
We call $\neighbor{v}{d}$ the \emph{$d$-neighbor}.
When $d=1$, we sometimes omit the subscript, i.e., $N(v) = N_1(v)$.
We also define $N_d(S)$ for a set $S \subseteq V$
as $N_d(S) = \bigcup_{v \in S}N_d(v)$.
In other words, $N_d(S)$ represents the set of vertices with distance at most $d$ from
at least one vertex in $S$.
The notations we will frequently use hereafter are summarized in Table~\ref{tbl:notation}.

\begin{table}[t]
\centering \small
\caption{Frequently used notations.}
\label{tbl:notation}
\begin{tabular}{l|l}
\toprule
\textbf{Notation} & \textbf{Description} \\
\midrule
\multicolumn{2}{c}{{\textit{(In the context of the box cover problem)}}} \\
$G = (V, E)$ & The graph. \\
$n, m$ & The numbers of vertices and edges in $G$. \\
$N_\delta(v)$ & The vertices with distance at most $\delta$ from $v$. \\
\midrule
\multicolumn{2}{c}{{\textit{(In the context of the set cover problem)}}} \\
$\set{S_p}_{p \in P}$ & The set family. \\
$n$ & The number of elements and collections. \\
\midrule
\multicolumn{2}{c}{{\textit{(Bottom-$k$ min-hash sketch)}}} \\
$k$ & The trade-off parameter of min-hash sketches.\\
$r_i$ & The rank of an item $i$. \\
$\widetilde{S}$ & The min-hash sketch of set $S$. \\
$\widetilde{C}(\widetilde{S})$ & The estimated cardinality of set $S$. \\
\bottomrule
\end{tabular}
\vspace{-1em}
\end{table}

\subsection{Bottom-{\secit{k}} Min-Hash Sketch}

In this subsection, we review the bottom-$k$ min-hash sketch and its cardinality estimator~\cite{ads/original,ads/cohen15}.
Let $X$ denote the ground set of items.
We first assign a random rank value $\rank{i} \sim U(0, 1)$ to each item $i \in X$,
where $U(0, 1)$ is the uniform distribution on $(0, 1)$.
Let $S$ be a subset of $X$.
For an integer $k \geq 1$, the bottom-$k$ min-hash sketch of $S$ is defined as $\widetilde{S}$, where
$i \in \widetilde{S} \iff \rank{i} \leq \kth\set{ \rank{j} \mid j \in S}.$
In other words, $\widetilde{S}$ is the set of vertices with the $k$ smallest rank values.
We define $\widetilde{S} = S$ if $\abs{S} < k$.

For a set $S \subseteq X$, the \emph{threshold rank} $\tau(S)$ is defined as follows.
If $\abs{S} \geq k$, $\tau(S) = k\textsf{-th}\set{ r_i \mid i \in S}$. Otherwise, $\tau(S) = (k-1) / \abs{S}$.
Note that $\tau(S) = \tau(\widetilde{S})$.
Using sketch $\widetilde{S}$, we estimate the cardinality $\abs{S}$ as
$\widetilde{C}(S) = (k - 1)/\tau(\widetilde{S})$.
Its relative error is theoretically bounded as follows.

\begin{lemma}[Bottom-$k$ cardinality estimator~\cite{ads/original,ads/cohen15}]
The cardinality estimation $\widetilde{C}(S)$ is an unbiased estimator of $\abs{S}$,
and has a coefficient of variation (CV)\footnote{The CV is the ratio of the standard deviation to the mean.} of at most $1 / \sqrt{k - 2}$.
\end{lemma}

The following corollary can be obtained by applying Chernoff bounds~\cite{ads/influence}.

\begin{corollary}
  \label{corollary:error}
  For $\epsilon>0$ and $c>1$,
  by setting $k \geq (2+c) \epsilon^{-2} \ln \abs{X}$,
  the probability of the estimation having a relative error larger than $\epsilon$ is at most $1/ \abs{X}^c$.
\end{corollary}

In addition, our algorithms
heavily rely on the \emph{mergeability} of min-hash sketches.
Suppose $S_1, S_2 \subseteq X$ and $S_3 = S_1 \cup S_2$.
Then, since $\widetilde{S}_3 \subseteq \widetilde{S}_1 \cup \widetilde{S}_2$,
$\widetilde{S}_3$ can be obtained only from $\widetilde{S}_1$ and $\widetilde{S}_2$.
We denote this procedure as \textsf{Merge-and-Purify}
(e.g., $\widetilde{S}_3 = \textsf{Merge-and-Purify}(\widetilde{S}_1, \widetilde{S}_2)$).

For simplicity, we assume that $r_i$ is unique for $i \in X$,
and sometimes identify $i$ with $r_i$.
In particular,
we use the comparison between elements
such as $i < j$ for $i, j \in X$, where we actually compare $r_i$ and $r_j$.
We also define $k\textsf{-th}(S)$ as the element with the $k$-th smallest rank in $S \subseteq X$.

\subsection{Problem Definition}\label{sec:problem}

\subsubsection{Graph Fractality}\label{sec:fractality}

The fractality of a network~\cite{Song2005} is a generalization of the fractality of a geometric object in Euclidean space~\cite{Falconer2003}.
A standard way to determine the fractality of a geometric object is to use the so-called box-counting method; we tile the object with cubes of a fixed length and count the number of cubes needed. If the number of cubes follows a power-law function of the cube length, the object is said to be fractal.
A fractal object holds a self-similar property so that we observe similar structure in it when we zoom in and out to it.

The idea of the box counting method is generalized to analyze the fractality of networks~\cite{Song2005}.
The box-covering method for a network works by covering the network with boxes of finite length $\ell$, which refers to a subset of vertices in which all vertices are within distance $\ell$.
For example, a box with $\ell = 1$ is a set of nodes all adjacent to each other.
If the number of boxes of length $\ell$ needed to cover the whole network, denoted by $b(\ell)$, follows a power-law function of $\ell$: $b(\ell) \propto \ell^{-d}$,
the network is said to be fractal. The exponent $d$ is called the fractal dimension.
As can be noticed, $b(\ell)$ crucially depend on how we put the boxes.
Theoretically, we have to put boxes such that $b(\ell)$ is minimized to assess its precise scaling.
However, this box-covering problem is NP-hard and that is why we propose our new approximation algorithm of this problem in the rest of this paper.

After computing $b(\ell)$ for a network, we want to decide whether the network is fractal or not.
A typical indicator of a non-fractal network is an exponential form: $b(\ell) \propto \exp(c \ell)$, where $c$ is a constant factor~\cite{Takemoto2014}.
Therefore, comparison of the fitting of the obtained $b(\ell)$ to power-law and exponential functions enables us to determine the fractality of the network.
Figure~\ref{fig:intro} illustrates the comparison for
{$(3,3,7)$-flower},
a network model with ground-truth fractality~\cite{Rozenfeld2007} (see Section~\ref{sec:procedure}).
Since $b(\ell)$ is closer to the power-law than to the exponential function in this case, the fitting procedure correctly indicates the fractality of this network model.

It should be noted that the fractality of a network suggests its self-similarity.
Let us aggregate the vertices in a box into a supervertex and then aggregate the edges spanning two boxes into a superedge.
Then we obtain a coarse-grained version of the original network.
If the original network is fractal, the vertex degree distributions of the original and coarse-grained networks are (statistically) the same~\cite{Song2005}.
Note that the fractality and self-similarity of a scale-free network are not equivalent, and a non-fractal scale-free network can be self-similar under certain conditions~\cite{Kim2007}.

\begin{figure}
\centering
\includegraphics[width=0.7\hsize]{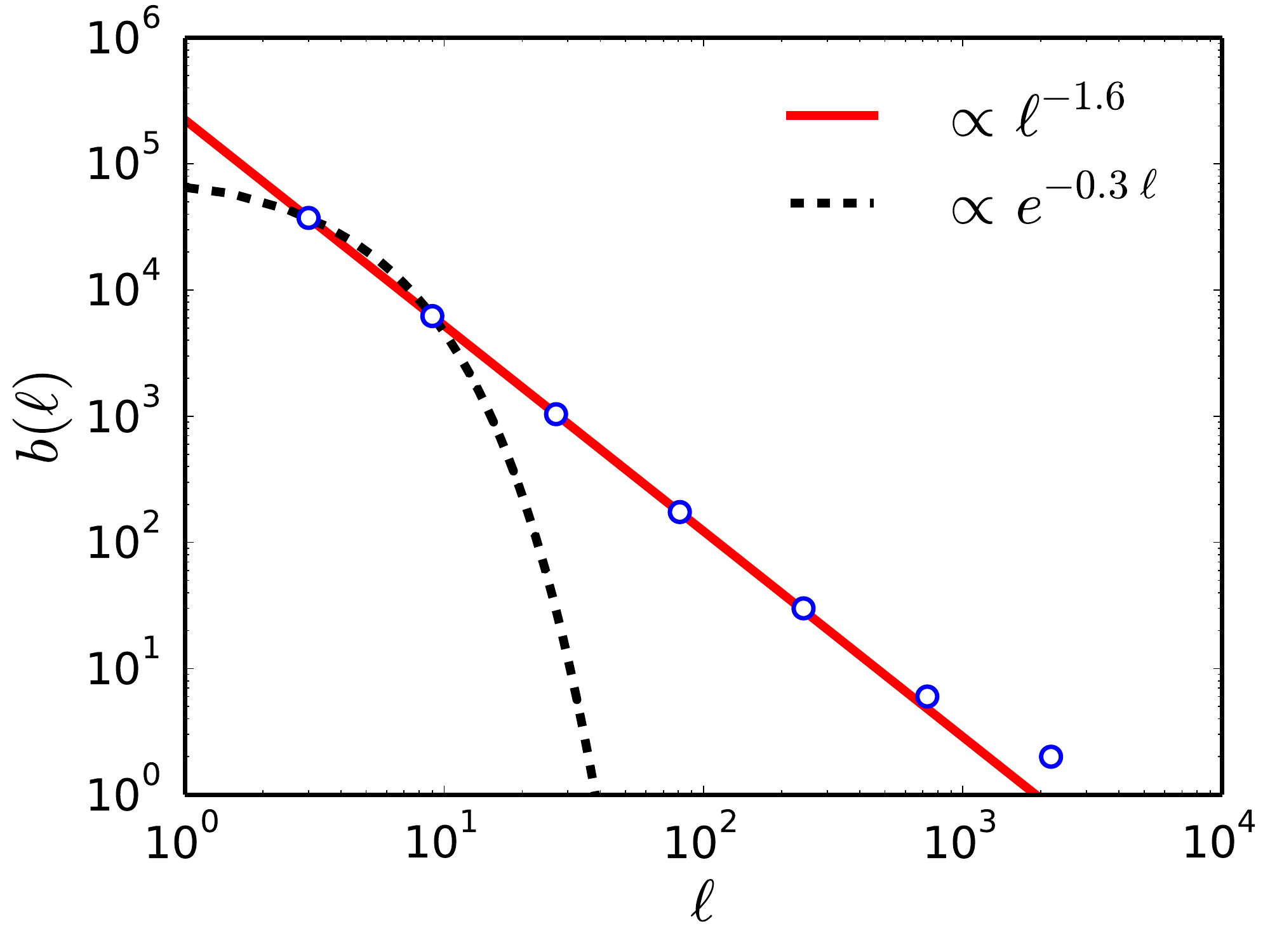}
\vspace{-1em}
\caption{Fitting of $b(\ell)$ to the power-law and exponential functions for
a synthetic fractal network.}
\vspace{-1em}
\label{fig:intro}
\end{figure}

\subsubsection{Box Cover}

As we described in the previous section, the fractality of graphs is analyzed by solving the box covering problem.
The problem has two slightly different versions:
the diameter version~\cite{Song2005} and
the radius version~\cite{Song2007}.
It has been empirically shown that these two versions yield negligible difference in the results.
In this study, we focus on the radius version, which is defined as follows.

\begin{problem}[\textsc{BoxCover}]
  \label{prb:box_cover}
In the \textsc{BoxCover} problem, given a graph $G$ and a radius limit $\ell > 0$,
the objective is to find a set $S \subseteq V$ of the minimum size
such that $N_\ell(S) = V$.
\end{problem}
The size of set $S$ is equal to $b(\ell)$ discussed in Section~\ref{sec:fractality}.
In this study, we consider a slightly relaxed variant of the \textsc{BoxCover} problem,
named $(1-\epsilon)$-\textsc{BoxCover}.
The $(1-\epsilon)$-\textsc{BoxCover} problem is defined as follows.

\begin{problem}[\textsc{$(1-\epsilon)$-BoxCover}]
  \label{prb:eps_box_cover}
  In the \textsc{$(1-\epsilon)$-BoxCover} problem, we are given a graph $G$, a radius limit $\ell > 0$
  and an error tolerance parameter $\epsilon > 0$.
The objective is to find a set $S \subseteq V$ of the minimum size
such that $\abs{N_\ell(S)} \geq (1-\epsilon)n$.
\end{problem}

\subsubsection{Set Cover}
The \textsc{BoxCover} problem is a special case of the \textsc{SetCover} problem,
which is defined as follows.

\begin{problem}[\textsc{SetCover}]
In the \textsc{SetCover} problem, we are given a set family $\set{S_p}_{p \in P}$.
The objective is to find a set $R \subseteq P$ of the minimum size
such that $\bigcup_{p \in R} S_p = \bigcup_{p \in P} S_p$.
\end{problem}

The proposed box-covering algorithm deals with a slightly different version of \textsc{SetCover},
named \emph{\textsc{$(1-\epsilon)$-SetCover} with sketched input} as a key subproblem,
which is defined as follows.

\begin{problem}[\textsc{$(1-\epsilon)$-SetCover} with sketched input]
\label{prb:eps_set_cover}
  In the sketched input version of the \textsc{$(1-\epsilon)$-SetCover} problem,
  we are given the min-hash sketches of a set family $\{\widetilde{S}_p\}_{p \in P}$
  and an error tolerance parameter $\epsilon > 0$.
The objective is to find a set $R \subseteq P$ of the minimum size
such that $| \bigcup_{p \in R} S_p | \geq (1-\epsilon) | \bigcup_{p \in P} S_p |$.
\end{problem}

We first design an efficient approximation algorithm for \textsc{$(1-\epsilon)$-SetCover}
(Section~\ref{sec:set_cover}).
We then propose a new box-covering algorithm using it (Section~\ref{sec:box_cover}).

\section{Set Cover in Sketch Space}
\label{sec:set_cover}

In this section, we design an efficient approximation algorithm
for the sketched-input version of \textsc{$(1-\epsilon)$-SetCover} (Problem~\ref{prb:eps_set_cover}).
We call each $p \in P$ a \emph{collection}
and $i \in S_p$ an \emph{element}.
Because of the connection to the \textsc{BoxCover} problems,
we assume that the numbers of collections and elements are equal.
We denote them by $n$,
that is, $\abs{P} = | \bigcup_{p \in P}S_p | = n $.
For $R \subseteq P$,
we define $S_R = \bigcup_{p \in R} S_p$.
Moreover, for simplicity,
we denote $\widetilde{C}(S_R)$ by $\widetilde{C}(R)$,
which can be calculated from merged min-hash sketch $\widetilde{S}_R$.

We first explain the basic greedy algorithm that runs in $O(n^2 k)$ time,
and then present its theoretical solution guarantee.
Finally, we propose an efficient greedy algorithm,
which runs in $O(n k \log n)$ time and produces the exact same solution as the basic algorithm.

\subsection{Basic Greedy Algorithm}
\label{sec:set_cover/basic}

Our basic greedy algorithm \textsf{Select-Greedily-Naive} is described as Algorithm~\ref{alg:greedy_basic}.
We start with an empty set $R = \set{}$.
In each iteration, we calculate $\widetilde{C}(R \cup \set{p})$ for every $p \in P \setminus R$,
and select $p$ that maximizes the estimated cardinality, and add it to $R$.
We repeat this until $\widetilde{C}(R)$ gets at least $(1 - \epsilon / 2) n$,
and the resulting $R$ is the solution.

\begin{algorithm}[t!]
\caption{\textsf{Select-Greedily-Naive}$(\{\widetilde{S}_p\}_{p \in P})$}
\small
\label{alg:greedy_basic}
\begin{algorithmic}[1]
  \setstretch{0.95}
  \State $R \gets \set{}, \widetilde{S}_R \gets \set{}$
  \While{$R \not= P$ and $\widetilde{C}(R) < (1 - \epsilon / 2)n$}  \label{alg:greedy_basic:loop}
    \State $p \gets \argmax \{ \widetilde{C}(R \cup \set{p}) \mid p \in P \setminus R \} $
    \State $R \gets R \cup \set{p}, \widetilde{S}_R \gets \textsf{Merge-and-Purify}(\widetilde{S}_R, \widetilde{S}_p)$
  \EndWhile
  \State \Return $R$
\end{algorithmic}
\end{algorithm}

To calculate $\widetilde{C}(R \cup \set{p})$,
together with $R$, we manage the merged min-hash sketch $\widetilde{S}_R$,
so that $\widetilde{S}_R$ always corresponds to the min-hash sketch of $S_R$.
To this end, we use the merger operation of min-hash sketch.
Let us assume that the items in a min-hash sketch are stored in the ascending order of their ranks.
Then, merging two min-hash sketches can be done in $O(k)$ time like in the merge sort algorithm;
we just need to pick the top-$k$ distinct items with the lowest ranks in the two min-hash sketches.
The complexity analysis of the algorithm is as follows.

\begin{lemma}
  Algorithm \textsf{Select-Greedily-Naive} runs in $O(n^2 k)$ time and $O(n k)$ space.
\end{lemma}

\begin{proof}[Proof Sketch]
This algorithm always terminates as, even in the worst case,
after the $n$-th iteration, $R$ gets $P$.
Each iteration takes $O(n k)$ time,
and the number of iterations is at most $n$.
Therefore, the time complexity is $O(n^2 k)$ time.
\end{proof}

\subsection{Theoretical Solution Guarantee}
\label{sec:set_cover/theory}

We can guarantee the quality of the solution produced by the above algorithm as follows.

\begin{lemma}
  \label{lemma:greedy_accuracy}
For $\epsilon \geq 2 \sqrt{5 (\ln n) / k}$,
algorithm \textsf{Select-Greedily-Naive} produces
a solution of \textsc{$(1-\epsilon)$-SetCover}
within a factor $1 + 2\ln n$ of the optimum for \textsc{SetCover}
with a probability of at least $1 - 1/n$.
\end{lemma}

In other words, with a high probability
that asymptotically approaches 1,
${R}$ is the solution of \textsc{$(1-\epsilon)$-SetCover}
and $| {R}| \leq  (1 + 2 \ln n) |{R}^*|$,
where ${R}$ is the output of algorithm \textsf{Select-Greedily-Naive}
and $R^*$ is the optimum solution of \textsc{SetCover} (with the same set family as the input).

\begin{proof}
  Let ${R}$ be the output of algorithm \textsf{Select-Greedily-Naive}, and
  $R^*$ be the optimum solution to \textsc{SetCover}.
  Let ${R}_0 = \emptyset$
  and ${R}_i \subseteq P$ be the currently selected sets after the $i$-th iteration
  of algorithm \textsf{Select-Greedily-Naive}.
  Let $C(R) = | \bigcup_{p \in R} S_p |$.
  From Corollary~\ref{corollary:error},
  for set $R \subseteq P$,
  the probability of $\widetilde{C}(R)$
  having a relative error larger than $\epsilon/2$ is at most $1/n^3$.

  At the $i$-th iteration, there is some collection $p$ such that
  \begin{equation} \label{eq:1}
  n - C({R}_i \cup \set{p}) \leq
  \left(1 - \frac1{|R^*|}\right)\left(n - C({R}_i)\right),
  \end{equation}
  since otherwise there would not be
  any solution ${R}^*$ of that size to \textsc{SetCover}.
  During the $i$-th iteration,
  there are at most $n$ new sets to be examined,
  and thus the union bound implies that
  the relative error between $C$ and $\widetilde{C}$ is at most $\epsilon/2$
  with a probability of at least $1 - 1/n^2$.
  Therefore, with that probability,
  \begin{equation} \label{eq:2}
  \widetilde{C}({R}_{i+1}) \geq
  \widetilde{C}({R}_i \cup \set{p}) \geq
  \left(1 - \frac{\epsilon}{2} \right){C}({R}_i \cup \set{p}).
  \end{equation}
  From Inequalities~\ref{eq:1} and \ref{eq:2},
  with some calculation, we have
  $$
  \left( 1- \frac{\epsilon}2 \right)n - \widetilde{C}({R}_{i+1})
  \leq \left(1 - \frac1{|R^*|}\right)
  \left( \left( 1 - \frac\epsilon2 \right)n - \widetilde{C}({R}_i)\right),
  $$
  with a probability of at least $1 - 1/n^2$.
  As the number of iterations is at most $n$,
  by applying the union bound over all iterations, we obtain
  $$
  \left(1- \frac{\epsilon}2 \right) n - \widetilde{C}({R}_{i})
  \leq \left(1 - \frac1{|R^*|}\right)^i n
  < e^{-\frac{i}{| {R}^* |}} n,
  $$
  with a probability of at least $1 - 1/n$.
  If $i$ is at least $2 | {R}^* | \ln n$,
  it becomes strictly less than $1/n$, which is smaller than
  the resolution of $\widetilde{C}$.
  Therefore, the number of iterations is at most
  $\lceil 2 | {R}^* | \ln n \rceil$, and thus $| {R}| \leq \lceil 2 | {R}^* | \ln n \rceil \leq  (1 + 2 \ln n) |{R}^*|$.
  Moreover, %
  $C({R}) \geq (1 - \epsilon/2) \widetilde{C}({R})
  \geq (1 - \epsilon)n$,
  and thus ${R}$ is the solution to the $(1-\epsilon)$-\textsc{SetCover} problem.
\end{proof}

\subsection{Near-Linear Time Greedy Algorithm}
\label{sec:set_cover/fast}

Algorithm \textsf{Select-Greedily-Naive} takes quadratic time,
which is unacceptable for large-scale set families.
Therefore, we then design an efficient greedy algorithm \textsf{Select-Greedily-Fast},
which produces the exact same output as algorithm \textsf{Select-Greedily-Naive}
but runs in $O(nk \log n)$ time.
As the input size is $O(nk)$, this algorithm is near-linear time.

\definecolor{mygray}{rgb}{0.4,0.4,0.4}

\begin{algorithm}[t!]
\caption{\textsf{Select-Greedily-Fast}$(\{\widetilde{S}_p\}_{p \in P})$}
\small
\label{alg:greedy_fast}
\begin{algorithmic}[1]
\setstretch{0.95}
\State \textcolor{mygray}{\small // Initialization}
\State $R \gets \set{}, \widetilde{S}_R \gets \set{}$
\State $Q_A \gets $ an empty min-queue {\footnotesize (key: ranks, value: collections)}
\State $Q_B \gets $ an empty min-queue {\footnotesize (key: integers, value: collections)}
\For{$j = 1, 2, \dots, k$}
  \State $T(j) \gets $ a binary search tree {\footnotesize(key: ranks, value: collections)}
\EndFor
\State $I_i \gets \{p \in P \mid i \in \widetilde{S}_p \}$
\ForAll{$p \in P$}
  \State Insert $(p, k\textsf{-th}\{r_i \mid i \in \widetilde{S}_p \})$ to $Q_A$ and $T(1)$
  \State $a_p \gets k, b_p \gets 0$
\EndFor
\State \textcolor{mygray}{\small // Main loop}
\While{$R \not= P$ and $\widetilde{C}(R) < (1 - \epsilon / 2)n$}  \label{alg:greedy_fast/loop}
  \State \textcolor{mygray}{\small // Selection}
  \State $p \gets \argmax \{ \widetilde{C}(R \cup \set{p}) \mid p \text{ is at the top of } Q_A \text{ or } Q_B \} $
  \State $R \gets R \cup \set{p}$
  \State Remove $p$ from $Q_A$, $Q_B$ and $T$
  \State $\widetilde{S}'_R \gets \textsf{Merge-and-Purify}(\widetilde{S}_R, \widetilde{S}_p)$ %
  \State $\Delta \gets \widetilde{S}'_R \setminus \widetilde{S}_R$,\,\,
    $\widetilde{S}_R \gets \widetilde{S}'_R$
  \State \textcolor{mygray}{\small // Notifying events of type 3}
  \ForAll{$i \in \Delta$}
    \ForAll{$p \in I_i$}
      \If{$p \in Q_B$}
        \State Move $p$ from $T({b_p})$ to $T({b_p + 1})$
        \ElsIf{$a_p\textsf{-th}(\widetilde{S}_p) = i$}
        \State \textbf{while} $a_p\textsf{-th}(\widetilde{S}_p) = i$ \textbf{do} $a_p \gets a_p - 1$
        \State Update $p$'s key in $T({b_p + 1})$ to $a_p\textsf{-th}(\widetilde{S}_p)$
        \State Remove $p$ from $Q_A$ and Insert $(p, b_p + 1)$ to $Q_B$
      \Else
        \State Move $p$ from $T({b_p + 1})$ to $T({b_p + 2})$
      \EndIf
      \State $b_p \gets b_p + 1$
    \EndFor
  \EndFor
  \State \textcolor{mygray}{\small // Notifying events of types 1 and 2-1}
  \For{$j = 1, 2, \ldots, k$}
  \State $P' \gets $ Retrieve those with keys $\geq j\textsf{-th}(\widetilde{S}_R)$ from $T(j)$
    \ForAll{$p \in P'$}
      \State $r \gets j\textsf{-th}(\widetilde{S}_R)$
      \State Remove $p$ from $Q_A$, $Q_B$ and $T(j)$
      \If{$p \in Q_A$}
        \State $a_p \gets a_p - 1, b_p \gets b_p + 1$
        \If{$a_p\textsf{-th}(\widetilde{S}_p) \in \widetilde{S}_R$}
          \State $b_p \gets b_p - 1, \, r \gets (j-1)\textsf{-th}(\widetilde{S}_R)$
        \EndIf
      \EndIf
      \State \textbf{while} $a_p\textsf{-th}(\widetilde{S}_p) \in \widetilde{S}_R$ \textbf{do} $a_p \gets a_p - 1$
      \If{$a_p\textsf{-th}(\widetilde{S}_p) > r $}
        \State Insert $(p, a_p\textsf{-th}(\widetilde{S}_p))$ to $Q_A$ and $T({b_p + 1})$
      \Else
        \State Insert $(p, b_p)$ to $Q_B$
        \State Insert $(p, a_p\textsf{-th}(\widetilde{S}_p))$ to $T({b_p})$
      \EndIf
    \EndFor
  \EndFor
\EndWhile
\State \Return $R$
\end{algorithmic}
\end{algorithm}

The behavior of \textsf{Select-Greedily-Fast}
at a high level is the same as that of \textsf{Select-Greedily-Naive}.
That is, we start with an empty set $R = \set{}$, and,
at each iteration, it adds $p \in P \setminus R$ with the maximum gain on $\widetilde{C}$ to $R$.
The central idea underlying the speed-up is
to classify the state of each $p \in P$ at each iteration into two types and
manage differently to reduce the reevaluation of the gain.
To this end, we closely look at the relation between sketches $\widetilde{S}_p$ and $\widetilde{S}_R$.

\newcommand{\SRp}{\widetilde{S}_{R \cup p}}
\newcommand{\SR}{\widetilde{S}_{R}}
\newcommand{\Sp}{\widetilde{S}_{p}}
\newcommand{\Sq}{\widetilde{S}_{q}}
\newcommand{\SRR}{\widetilde{S}_{R'}}
\newcommand{\SRRp}{\widetilde{S}_{R' \cup p}}

\newcommand{\caseparagraph}[1]{\vspace{0.5em} \noindent \textit{\underline{#1:} }}

\myparagraph{Types and Variables}
Let us assume that we are in the main loop of the greedy algorithm.
Here, we have a currently incomplete solution $R \subset P$.
Let $p \in P \setminus R$.
We define that $p$ belongs to \emph{\underline{type A}} if
the $k$-th element of $\SRp$ is in $S_p$, i.e., $k\textsf{-th}(\SRp) \in S_p$.
Otherwise, $p$ is \emph{\underline{type B}}.
Note that $\SRp = \textsf{Merge-and-Purify} \allowbreak (\SR, \allowbreak \Sp)$.

We define %
$a_p = |{\SRp \cap \Sp}|$, $b_p = |{\SRp \cap \SR}|$, and
$c_p = |{\SR \cap \Sp}|$.
Please note that, if $p$ is a type-A collection,
$\widetilde{C}(R \cup \set{p})$
is determined by $k\textsf{-th}(\SRp) = a_p\textsf{-th}(\Sp)$.
Similarly, if $p$ is a type-B collection,
$\widetilde{C}(R \cup \set{p})$
is determined by $k\textsf{-th}(\SRp) = b_p\textsf{-th}(\SR)$.

\myparagraph{Events to be Captured}
\label{sec:set_cover/fast/event}
Suppose that we have decided to adopt a new collection
and $R$ is about to be updated to $R'$ (i.e., $R' = R \cup \set{p'}$ for some collection $p' \in P$).
Let us first assume that
a single element appeared in the merged sketch,
i.e., $\SRR \setminus \SR = \set{i}$.
Let $p \in P \setminus R'$.
In the following, we examine and classify the \emph{events}
where the evaluation of $p$ is updated,
i.e., $\widetilde{C}(R \cup \set{p}) \not= \widetilde{C}(R' \cup \set{p})$
(types 1 and 2),
or $c_p$ is updated (type 3).

\caseparagraph{Type 1}
We assume that $i \not\in \Sp$ and $p$ is type A.
From the definition,
$\tau(\SRp) = a_p\textsf{-th}(\Sp)$,
and, $b_p\textsf{-th}(\SR) \leq a_p\textsf{-th}(\Sp) \allowbreak
< (b_p+1)\textsf{-th}(\SR)$.
Therefore, $\tau(\SRp) \not= \tau(\SRRp)$
if and only if $(b_p+1)\textsf{-th}(\SRR) \not= (b_p+1)\textsf{-th}(\SR)$
and $(b_p+1)\textsf{-th}(\SRR) \allowbreak  < a_p\textsf{-th}(\Sp)$.
We define that a \emph{type-1 event} happens to $p$ when this condition holds.

\caseparagraph{Type 2}
Similarly, we assume that $i \not\in \Sp$ and $p$ is type B.
From the definition, $\tau(\SRp) = b_p\textsf{-th}(\SR)$
and $a_p\textsf{-th}(\Sp) < b_p\textsf{-th}(\SR)$.
Thus $\tau(\SRp) \not= \tau(\SRRp)$,
if and only if $b_p\textsf{-th}(\SR) \allowbreak \not= b_p\textsf{-th}(\SRR)$.
There are two cases:
$b_p\textsf{-th}(\SRR) \leq a_p\textsf{-th}(\Sp)$
\textit{(\underline{type 2-1})},
after which $p$ becomes type A,
or
$b_p\textsf{-th}(\SRR) > a_p\textsf{-th}(\Sp)$
\textit{(\underline{type 2-2})},
after which $p$ still belongs to type B.

\caseparagraph{Type 3}
If $i \in \Sp$, then $c_p$ will be incremented.

\vspace{0.5em}
The following lemma is the key to the efficiency of our algorithm.

\begin{lemma}
  \label{lemma:3k}
  For each $p \in P$, throughout the algorithm execution,
  events of type 1, type 2-1, or type 3
  occur at most $3k$ times in total.
\end{lemma}

\begin{proof}[Proof Sketch]
  We use the progress indicator $\Phi = k - a_p + b_p + c_p$.
  Initially, $a_p = k$ and $b_p = c_p = 0$;
  hence $\Phi = 0$.
  For each event occurrence, $\Phi$ increases by at least one.
  As $a_p \geq 0$ and $b_p, c_p \leq k$,
  $\Phi \leq 3k$. %
\end{proof}

Please note that events of type 2-2 are not considered in the above lemma,
and, indeed, they happen $\Theta(n)$ times in the worst case for each collection.
Therefore, we design the algorithm
so that we do not need to capture type-2-2 events.

\myparagraph{Finding the Maximum Gain}
To adopt a new collection in each iteration,
we need to efficiently find the collection that gives the maximum gain.
We clarify the ordering relation in each type.

\caseparagraph{Type A}
Let $p, q \in P \setminus R$ be type-A collections,
then $\widetilde{C}(R \cup \set{p}) \geq \widetilde{C}(R \cup \set{q})$
if and only if
$a_p\textsf{-th}(\Sp) \leq a_q\textsf{-th}(\Sq)$.

\caseparagraph{Type B}
Let $p, q \in P \setminus R$ be type-B collections,
then $\widetilde{C}(R \cup \set{p}) \geq \widetilde{C}(R \cup \set{q})$
if and only if
$b_p\textsf{-th}(\SR) \leq b_q\textsf{-th}(\SR)$,
which is equivalent to $b_p \leq b_q$.

\myparagraph{Data Structures}
We use the following data structures to notify the
collections about an event occurrence.

\caseparagraph{Type 1}
For each type-A collection $p$,
as we observed above,
$p$ wants to be notified about a type-1 event
when $(b_p+1)\textsf{-th}(\SR)$ becomes smaller than $a_p\textsf{-th}(\Sp)$.
Therefore, for $j = 1, 2, \ldots, k$,
we prepare a binary search tree $T(j)$,
where values are collections and keys are ranks
(i.e., collections are managed in the ascending order of ranks in each tree).
For each type-A collection $p$,
we put $p$ in $T(b_p + 1)$ with key $a_p\textsf{-th}(\Sp)$.
Then, when $j\textsf{-th}(\SR)$ is updated to a new value,
from $T(j)$, we retrieve collections with keys larger than or equal to the new value
and notify them about an event.

\caseparagraph{Type 2-1}
Similarly, for each type-B collection $p$,
$p$ wants to be notified about a type-2-1 event
when $b_p\textsf{-th}(\SR)$ becomes smaller than or equal to $a_p\textsf{-th}(\Sp)$.
Thus, we store $p$ in $T(b_p)$ and set its key to $a_p\textsf{-th}(\Sp)$.
Then, when $j\textsf{-th}(\SR)$ is updated to a new value,
we retrieve those in $T(j)$ with keys larger than or equal to the new value
and notify them about an event.

\caseparagraph{Type 3}
To capture type-3 events,
the use of an inverted index suffices.
That is, for each $i \in X$,
we precompute $I_i = \{p \in P \mid i \in \Sp \}$.
When $i$ comes to $\SR$, we notify the collections in $I_i$.

\vspace{0.5em}Moreover, we also need data structures to
find the collections with the maximum gain as follows.

\caseparagraph{Type A}
Type-A collections are managed in a minimum-oriented priority queue,
where the key of a collection $p$ is $a_p\textsf{-th}(\Sp)$.

\caseparagraph{Type B}
Type-B collections are managed in another minimum-oriented priority queue,
where the key of a collection $p$ is $b_p$.

\myparagraph{Overall Set-Cover Algorithm}
The overall algorithm of \textsf{Select-Greedily-Fast} is described as Algorithm~\ref{alg:greedy_fast}.
In each iteration, we adopt the new collection with maximum gain,
which can be identified by comparing the top elements of the two priority queues.
Then, we process events to update variables and data structures.
At the beginning of Section~\ref{sec:set_cover/fast/event},
we assumed that a single element appears in the new sketch.
When more than one elements come to the new sketch,
we basically process each of them separately.
See Algorithm~\ref{alg:greedy_fast} for the details of the update procedure.
The algorithm complexity and solution quality are guaranteed as follows.

\begin{lemma}
  \label{lemma:greedy_fast}
  Algorithm \textsf{Select-Greedily-Fast}
  runs in $O(nk \log n)$ time and $O(nk)$ space.
\end{lemma}

\begin{proof}[Proof Sketch]
  Each data structure operation takes $O(\log n)$ time,
  which, from Lemma~\ref{lemma:3k}, happens at most $3k$ times for each collection.
\end{proof}

\begin{lemma}
  \label{lemma:greedy_fast_correct}
  Algorithm \textsf{Select-Greedily-Fast}
  produces the same solution as algorithm \textsf{Select-Greedily-Naive}.
\end{lemma}

\begin{proof}[Proof Sketch]
  Both algorithms choose the collection with the maximum gain
  in each iteration.
\end{proof}

\section{Sketch-Based Box Covering}
\label{sec:box_cover}

In this section,
we complete our sketch-based box-covering algorithm
for the $(1-\epsilon)$-\textsc{BoxCover} problem (Problem~\ref{prb:eps_box_cover}).
We first propose an efficient algorithm
to construct min-hash sketches representing the $\ell$-neighbors,
and then present and analyze the overall box-covering algorithm,

\subsection{Sketch Generation}
\label{sec:box_cover:gen}

For $v \in V$,
we denote the min-hash sketch of $\neighbor{v}{\ell}$
as $\sketch{v}{\ell}$.
Here, we construct $\sketch{v}{\ell}$ for all vertices $v \in V$
to reduce the
$(1-\epsilon)$-\textsc{BoxCover} problem to
the $(1-\epsilon)$-\textsc{SetCover} problem (Problem~\ref{prb:eps_set_cover}).
Our sketch construction algorithm \textsf{Build-Sketches}
is described as Algorithm~\ref{alg:sketch}.

It receives a graph $G$ and a radius parameter $\ell$.
Each vertex $v$ manages a tentative min-hash sketch $X_v$.
Initially, $X_v$ only includes the vertex itself, i.e., $X_v = \set{v}$,
which corresponds to $\widetilde{N}_0(v)$.
Then, we repeat the following procedure for $\ell$ times
so that, after the $i$-th iteration, $X_v = \widetilde{N}_i(v)$.
This algorithm has a similar flavor to
algorithms for
approximated neighborhood functions and all-distances sketches~\cite{ads/anf, ads/hyper_anf,ads/cohen15}.

In each iteration,
for each vertex,
we essentially merge the sketches of its neighbors into its sketch
in a message-passing-like manner.
Two speed-up techniques are employed here to avoid an unnecessary insertion check.
For $v \in V$, let $A_v$ be the vertices
in whose sketches $v$ is added to in the last iteration.
First, for each $v \in V$,
we try to insert $v$
only into the sketches of the vertices that are neighbors of $A_v$,
as $v$ cannot be inserted into other vertices.
Second, we conduct the procedure above in the increasing order of ranks,
since this decreases the unnecessary insertion.
We prove its correctness
and complexity as follows.

\begin{algorithm}[t!]
\small
\caption{\textsf{Build-Sketches}$(G, \ell)$}
\label{alg:sketch}
\begin{algorithmic}[1]
  \setstretch{0.95}
  \State $X_v \gets \set{v}$ \textbf{for all} $v \in V$.
  \For{$\ell$ times}
    \ForAll{$v \in V$ in the increasing order of $r_v$}
      \State $A_v \gets \set{u \in V \mid \text{$v$ is added to $X_u$ in the last iteration}}$
      \ForAll{$w \in N(A_v)$}
        \State $X_w \gets \textsf{Merge-and-Purify}(X_w, \set{v})$ \label{alg:sketch:insert}
      \EndFor
    \EndFor
    \If{$X_v$ was not modified for any $v \in V$}
      \State \textbf{break}
    \EndIf
  \EndFor
  \State \Return $\set{X_v}_{v \in V}$
\end{algorithmic}
\end{algorithm}

\begin{lemma}
  \label{lemma:sketch_correct_lemma}
  In algorithm \textsf{Build-Sketches},
  after the $i$-th iteration,
  $X_v = \widetilde{N}_i(v)$ for all $v \in V$.
\end{lemma}

\begin{proof}[Proof Sketch]
  We prove the lemma by mathematical induction on $i$.
  Since $\set{v} = \widetilde{N}_0(v)$, it is true for $i = 0$.
  Now we assume it holds for $i$
and prove it also holds for $i+1$.
Let $B = \set{u \in V \mid (v, u) \in E}$.
Since ${N}_{i+1}(v) = \set{v} \cup \bigcup_{u \in B} {N}_{i}(u)$,
and $\set{v} \in N_{i}(v) \subseteq N_{i+1}(v)$,
$\widetilde{N}_{i+1}(v)$ can be obtained
by merging $\widetilde{N}_i(u)$ for all $u \in B \cup \set{v}$.
\end{proof}

\begin{corollary}
  \label{lemma:sketch_correct}
Algorithm \textsf{Build-Sketches}
computes $\sketch{v}{\ell}$ for all $v \in V$.
\end{corollary}

\begin{lemma}
  \label{lemma:sketch_complexity}
  Algorithm \textsf{Build-Sketches} runs in $O((n + m)k \allowbreak \log k \allowbreak \min\set{\ell, \log n})$ expected time
  and $O(nk + m)$ space.
\end{lemma}

\begin{proof}[Proof Sketch]
  In addition to the graph,
  the algorithm stores a sketch of size $k$ for each vertex,
  and hence it works in $O(nk + m)$ space.
  Each insertion trial takes $O(\log k)$ time (Line~\ref{alg:sketch:insert}).
  Therefore, it suffices to prove that the number of traversed edges
  is $O((n + m)k \ell)$ and $O((n + m)k \log n)$.
  The former bound is easier,
  since, in each iteration,
  the number of last inserted elements in each sketch is at most $k$,
  and thus we traverse each edge at most $k$ times.

  For the latter bound,
  we count the expected number of vertices
  that are inserted once into $X_v$ for a vertex $v \in V$.
  The vertex that is $i$-th to arrive at $v$ is inserted into $X_v$ with a probability of $\min\set{1, k/i}$,
  and thus it is at most
  $$
  \sum_{i=1}^n \min\set{1, \frac{k}{i}} = k + k(H(n) - H(k)) = O(k \log n),
  $$
  where $H(i)$ is the $i$-th Harmonic number.
  Therefore, each edge is traversed at most $O(k \log n)$ times in total.
\end{proof}

\subsection{Overall Box-Cover Algorithm}
\label{sec:box_cover:overall}

The overall box-covering algorithm \textsf{Sketch-Box-Cover}
is as follows.
We first construct the min-hash sketches using algorithm \textsf{Build-Sketches}
and then solve the set cover problem in the sketch space using algorithm \textsf{Select-Greedily-Fast}.
The guarantees on performance and accuracy
are immediate from the previous lemmas and corollaries as below.

\begin{theorem}[Scalability guarantee]
  \label{theorem:performance}
  Algorithm \textsf{Sketch-Box-Cover}
  works in $O((n + m) k \log k \min\set{\ell, \log n})$ time
  and $O(nk \allowbreak + m)$ space.
\end{theorem}

\begin{theorem}[Solution accuracy guarantee]
  \label{theorem:accuracy}
With a probability of at least $1 - 1/n$,
for $\epsilon \geq 2 \sqrt{5 (\ln n) / k}$,
algorithm \textsf{Sketch-Box-Cover} produces
a solution to the \textsc{$(1-\epsilon)$-BoxCover} problem
within a factor $1 + 2\ln n$ of the optimum for the \textsc{BoxCover} problem.
\end{theorem}

Assuming $k$ is a constant,
the time and space complexities are near-linear.
Similarly, given a constant $\epsilon$,
the time and space complexities are still near-linear,
since it suffices to set $k = \lceil 20 \epsilon^{-2} \ln n \rceil$.
In practice, as seen in our experiments,
the algorithm produces solutions
that are much closer to the optimum than
what is expected from the above approximation ratio
with much smaller $k$.

\section{Practical Improvement}
\label{sec:improve}
In this section, we propose techniques
to improve the practicality of the proposed method.

\myparagraph{Exact Coverage Management}
For the termination condition in the greedy selection algorithm
(i.e., Line~\ref{alg:greedy_fast/loop} in Algorithm~\ref{alg:greedy_fast}),
when applied to the box cover problem,
we propose to use the exact coverage $C(R)$
instead of the estimated coverage $\widetilde{C}(R)$.
This technique makes the results more stable.
We can efficiently manage the exact coverage as follows.

First, we prepare an array $\delta$,
and initialize it as $\delta[v] = \infty$ for all $v \in V$.
After selecting a vertex $v$ in each iteration,
we conduct a pruned breadth-first search (BFS) from $v$.
Suppose we are visiting vertex $u$ with distance $d$ in this BFS.
If $\delta[u] \leq d$, then we \emph{prune} this BFS,
i.e., we do not traverse the edges from $u$.
Otherwise, we set $\delta[u] = d$ and continue the search.
We do not visit vertices with a distance larger than $\ell$.
The number of covered vertices is the number of non-infinity values in array $\delta$.
Since the value of $\delta[u]$ changes at most $\ell + 1$ times,
each vertex or edge is visited $O(\ell)$ times.
Therefore, the total time consumption of this process throughout all iterations
is $O((n + m) \ell)$.

\myparagraph{Multi-Pass Execution}
On the basis of the above exact coverage management technique,
we sometimes detect that,
even while the estimated coverage is saturated
(i.e., $\widetilde{C}(R) = \widetilde{C}(P)$),
the actual coverage is below the specified threshold.
In that case, to choose more vertices,
we propose to repeat the algorithm from sketch construction
until the actual coverage becomes higher than the threshold.

In the $i$-th pass,
we only care for vertices that are not covered by the previous passes.
This can be easily realized by modifying the algorithm \textsf{Build-Sketches}
so that, at Line 1, we set $X_v = \emptyset$ for already covered vertices.
For accurate results, node ranks should be reassigned for each pass.

\myparagraph{Exact Neighborhood}
To further improve the accuracy,
we propose to combine our sketch-based algorithm with a non-sketch-based algorithm.
For a very small radius parameter $\ell$,
neighborhood $N_\ell(v)$ is sometimes much smaller than $k$.
Moreover, even for a larger $\ell$,
when using the above multi-pass execution technique,
the remaining neighbors may become small in later passes.
In these cases, the sketching approach has little advantage.
Therefore, we detect such circumstances
and switch to a non-sketch-based greedy algorithm.
Interestingly, this switching can be done seamlessly.
If $| \widetilde{N}_\ell(v) | \leq k$,
then $\widetilde{N}_\ell(v) = {N_\ell}(v)$.
Therefore,
under such circumstances,
the output of algorithm \textsf{Build-Sketches}
can be immediately given to the non-sketch-based greedy algorithm.

The proposed overall procedure is as follows.
We specify a parameter $\alpha$.
We start by constructing the ``sketches'' with algorithm \textsf{Build-Sketches},
but, at first, we apply the algorithm as if  $k=\infty$,
i.e., we do not conduct purification on the min-hash sketches.
During the construction, if the total number of elements in all min-hash sketches
exceeds $\alpha n k$ at some point,
then we conduct purification on all the min-hash sketches,
continue the construction with the actual $k$ value,
and pass the resulting sketches to the sketch-based greedy algorithm.
Otherwise, we apply the non-sketch-based greedy algorithm to the resulting
``sketches,'' which are actually exact neighborhood sets.
Assuming parameter $\alpha$ is a constant,
the total time and space complexity remain the same.

\myparagraph{Exact Box Covering}
\label{sec:improve/box_cover}
Together with the preceding three techniques,
to further make the results reliable,
we propose to use our algorithm
for solving the original \textsc{BoxCover} problem (Problem~\ref{prb:box_cover})
rather than $(1-\epsilon)$-\textsc{BoxCover} problem (Problem~\ref{prb:eps_box_cover}).
In other words, we recommend setting $\epsilon=0$
to ensure that all vertices are completely covered.
As we will see in the experiments,
even with this seemingly extreme threshold, thanks to the above techniques,
both the running time and the solution quality are reasonable.

\section{Experiments}

In this section, we present our experimental results
to verify the performance of the algorithm.
Specifically, we compare our algorithm with other preceding algorithms in terms of accuracy and computation time.

We mainly focus on model networks, instead of empirical ones, in order to validate the results of our algorithm with ground-truth theoretical solutions and to investigate the scalability of the algorithm for various sizes of networks.
However, we also demonstrate the practicalness of our algorithm by applying it to a real million-scale web graph.
On the basis of the result, we reveal the fractality of such large-scale real graph for the first time.

\label{sec:experiments}
\subsection{Setup}\label{sec:procedure}

\myparagraphF{Environment}
Experiments were conducted on a Linux server with Intel Xeon X5650 (2.67 GHz) and 96GB of main memory.
Algorithms were implemented in C++ and compiled by gcc 4.8.4 with \texttt{-O3} option.

\myparagraph{Algorithms}
For comparison, we used a naive algorithms named greedy coloring (GC) and three advanced and popular algorithms, named maximum excluded mass burning (MEMB), minimal value burning (MVB), and compact box burning (CBB).
GC, MEMB, and CBB were introduced in \cite{Song2007} and MVB was in \cite{Schneider2012}.

\myparagraph{Network Models}
We used two network models with ground-truth fractality: the $(u,v)$-flower~\cite{Rozenfeld2007} and the Song-Havlin-Makse (SHM)~\cite{Song2006} model.
These models have power-law degree distributions, the representative characteristic of complex networks.
Both models can be either fractal or non-fractal, depending on the structural parameter values.
We refer to them as the $(u, v, g)$-flower and $(c,e, g)$-SHM model to indicate the parameter settings.
The common parameter $g$ $(g= 1,2,3,\dots)$ determines the network size $n$: $n=(w-2/w-1)w^g + w/w-1$, where $w \equiv u+v$ for the $(u, v, g)$-flower, and $n = (2c+1)^g n_0$ for the $(c,e,g)$-SHM model.
In addition to the flower and SHM models, we considered the Barab\'{a}si-Albert (BA) network model~\cite{Barabasi1999} as one of the most famous models of complex networks.
The BA model is not fractal~\cite{Song2005}. We refer to this model as $(c, t)$-BA, where $c$ is the number of edges that a new node has and $t$ sets the network size as $n = 125 \times 2^t$.

\myparagraph{Fractality Decision Procedure}
After the computation of the box-covering algorithms, we determined whether the obtained $b(\ell)$ indicates the fractality or not.
This task was done by fitting the $b(\ell)$ curve with a power-law function (i.e., fractal) and an exponential function (i.e., non-fractal) by using \texttt{optimize.leastsq} function in \texttt{SciPy} package of Python.
We used the parameters estimated by fitting the curves to linearized models as the initial values for the nonlinear fitting.
The key quantity was the ratio between the residual error of fitting to a power-law function and that to an exponential function, denoted by $r_{\rm fit}$.
If $- \log_{10} r_{\rm fit}$ is postive (i.e., $r_{\rm fit} < 1$), the network was supposed to be fractal. Otherwise, it was supposed to be non-fractal.
This procedure of fitting and comparison follows that used in \cite{Takemoto2014}.

{
  \tabcolsep=2.5mm
  \renewcommand{\arraystretch}{0.8}
\begin{table*}[htb]
\centering \small
\caption{Running time in seconds (\emph{Time}) and the relative error ratio of a power-law function, $- \log_{10} r_{\rm fit}$ (\emph{Fit}).
DNF means that it did not finish in one day or ran out of memory.}
\label{tbl:main_table}
\begin{tabular}{lrr|rr|rr|rr|rr|rr}
\toprule
\multicolumn{3}{c|}{\textbf{Graph}} & \multicolumn{2}{c|}{\textbf{Sketch}} & \multicolumn{2}{c|}{\textbf{MEMB~\cite{Song2007}}} & \multicolumn{2}{c|}{\textbf{GC~\cite{Song2007}}} & \multicolumn{2}{c|}{\textbf{MVB~\cite{Schneider2012}}} & \multicolumn{2}{c}{\textbf{CBB~\cite{Song2007}}} \\ %
Model & $\abs{V}$ & $\abs{E}$ & Time & Fit & Time & Fit & Time & Fit & Time & Fit & Time & Fit \\ \midrule
\multicolumn{13}{l}{$\bigtriangledown$ \textit{\textbf{Networks with ground-truth fractality}}
{(``Fit'' values are expected to be positive.)}} \\
\midrule
(2, 2, 4)-flower & 172 & 256 & 0 & 0.8 & 0 & 1.0 & 0 & 28.7 & 199 & 1.0 & 0 & 28.0 \\
(2, 2, 7)-flower & 10,924 & 16,384 & 15 & 2.5 & 10 & 3.4 & 228 & 27.7 & DNF & --- & 122 & 27.4 \\
(2, 2, 10)-flower & 699,052 & 1,048,576 & 8,628 & 3.5 & DNF & --- & DNF & --- & DNF & --- & DNF & --- \\
(2, 2, 11)-flower & 2,796,204 & 4,194,304 & 62,138 & 4.0 & DNF & --- & DNF & --- & DNF & --- & DNF & --- \\
(2, 3, 6)-flower & 11,720 & 15,625 & 26 & 1.2 & 14 & 1.1 & 146 & 0.1 & DNF & --- & 5,593 & 0.5 \\
(2, 3, 7)-flower & 58,595 & 78,125 & 286 & 1.1 & 377 & 1.0 & 8,538 & 0.1 & DNF & --- & DNF & --- \\
(2, 3, 8)-flower & 292,970 & 390,625 & 2,913 & 1.0 & DNF & --- & DNF & --- & DNF & --- & DNF & --- \\
(2, 4, 6)-flower & 37,326 & 46,656 & 138 & 1.1 & 121 & 1.0 & 2,422 & 1.7 & DNF & --- & 2,559 & 0.6 \\
(2, 4, 7)-flower & 223,950 & 279,936 & 1,526 & 1.0 & DNF & --- & DNF & --- & DNF & --- & DNF & --- \\
(3, 3, 6)-flower & 37,326 & 46,656 & 148 & 1.1 & 101 & 1.3 & 10,751 & 2.1 & DNF & --- & 1,284 & 1.4 \\
(3, 3, 7)-flower & 223,950 & 279,936 & 1,779 & 1.2 & DNF & --- & DNF & --- & DNF & --- & 61,562 & 1.5 \\
(3, 4, 5)-flower & 14,007 & 16,807 & 34 & 0.7 & 16 & 0.9 & 560 & 0.1 & DNF & --- & 3,380 & -0.4 \\
(3, 4, 7)-flower & 686,287 & 823,543 & 8,873 & 0.8 & DNF & --- & DNF & --- & DNF & --- & DNF & --- \\
(2, 0, 6)-SHM & 12,501 & 12,500 & 33 & 1.2 & 8 & 1.1 & 872 & 1.1 & 32 & 1.1 & 325 & 0.7 \\
(2, 0, 7)-SHM & 62,501 & 62,500 & 224 & 1.2 & 206 & 1.1 & 48,116 & 1.1 & 1,126 & 1.1 & 6,579 & 0.9 \\
(2, 0, 8)-SHM & 312,501 & 312,500 & 2,728 & 1.1 & DNF & --- & DNF & --- & DNF & --- & DNF & --- \\
(3, 0, 6)-SHM & 67,229 & 67,228 & 207 & 1.0 & 190 & 0.9 & 21,726 & 0.9 & 628 & 0.9 & 4,623 & 0.9 \\
\midrule
\multicolumn{13}{l}{$\bigtriangledown$ \textit{\textbf{Networks with ground-truth non-fractality}}
{(``Fit'' values are expected to be negative.)}} \\
\midrule
(1, 2, 10)-flower & 29,526 & 59,049 & 108 & -2.9 & 197 & -2.9 & 286 & -5.4 & 364 & -2.2 & 21,833 & -2.6 \\
(1, 2, 11)-flower & 88,575 & 177,147 & 466 & -3.8 & 1,641 & -3.8 & 2,999 & -6.2 & 3,610 & -2.6 & DNF & --- \\
(1, 2, 12)-flower & 265,722 & 531,441 & 1,774 & -4.6 & DNF & --- & 38,278 & -7.0 & DNF & --- & DNF & --- \\
(1, 3, 7)-flower & 10,924 & 16,384 & 20 & -3.0 & 16 & -2.7 & 44 & -4.8 & DNF & --- & 1,862 & -3.3 \\
(1, 3, 8)-flower & 43,692 & 65,536 & 123 & -4.7 & 280 & -3.2 & 826 & -6.0 & DNF & --- & 61,953 & -4.4 \\
(1, 3, 9)-flower & 699,052 & 1,048,576 & 4,195 & -6.0 & DNF & --- & DNF & --- & DNF & --- & DNF & --- \\
(1, 4, 6)-flower & 11,720 & 15,625 & 23 & -1.0 & 20 & -0.6 & 53 & -1.8 & DNF & --- & 3,781 & -1.9 \\
(1, 4, 7)-flower & 58,595 & 78,125 & 223 & -0.8 & 548 & -0.7 & 1,598 & -1.8 & DNF & --- & DNF & --- \\
(1, 4, 8)-flower & 292,970 & 390,625 & 1,678 & -0.7 & DNF & --- & 67,866 & -1.9 & DNF & --- & DNF & --- \\
(2, 1, 6)-SHM & 24,885 & 31,104 & 31 & -3.5 & 32 & -3.5 & 433 & -0.4 & 126 & -3.5 & 7,129 & -2.5 \\
(2, 1, 7)-SHM & 149,301 & 186,624 & 390 & -4.9 & 1,397 & -4.9 & 17,703 & -0.4 & 8,615 & -4.9 & DNF & --- \\
(3, 1, 5)-SHM & 14,045 & 16,384 & 12 & -2.6 & 8 & -2.6 & 97 & -0.3 & 25 & -2.6 & 1,224 & -2.9 \\
(3, 1, 6)-SHM & 112,349 & 131,072 & 210 & -4.1 & 580 & -4.2 & 9,504 & -0.3 & 2,070 & -4.2 & DNF & --- \\
(2, 1)-BA & 250 & 497 & 0 & -0.9 & 0 & -0.9 & 0 & -0.6 & 54 & -0.5 & 0 & -0.3 \\
(2, 4)-BA & 2,000 & 3,997 & 1 & -2.7 & 0 & -2.0 & 2 & -0.6 & DNF & --- & 404 & -0.1 \\
(2, 7)-BA & 16,000 & 31,997 & 17 & -1.3 & 76 & -1.3 & 154 & -0.6 & DNF & --- & DNF & --- \\
(2, 10)-BA & 128,000 & 255,997 & 377 & -1.5 & 3,535 & -1.5 & 12,457 & -0.6 & DNF & --- & DNF & --- \\
(2, 13)-BA & 1,024,000 & 2,047,997 & 6,474 & -1.4 & DNF & --- & DNF & --- & DNF & --- & DNF & --- \\
(2, 15)-BA & 4,096,000 & 8,191,997 & 36,125 & -1.4 & DNF & --- & DNF & --- & DNF & --- & DNF & --- \\
\bottomrule
\end{tabular}
\end{table*}
}

\begin{figure}
\centering
\includegraphics[width=0.48\hsize]{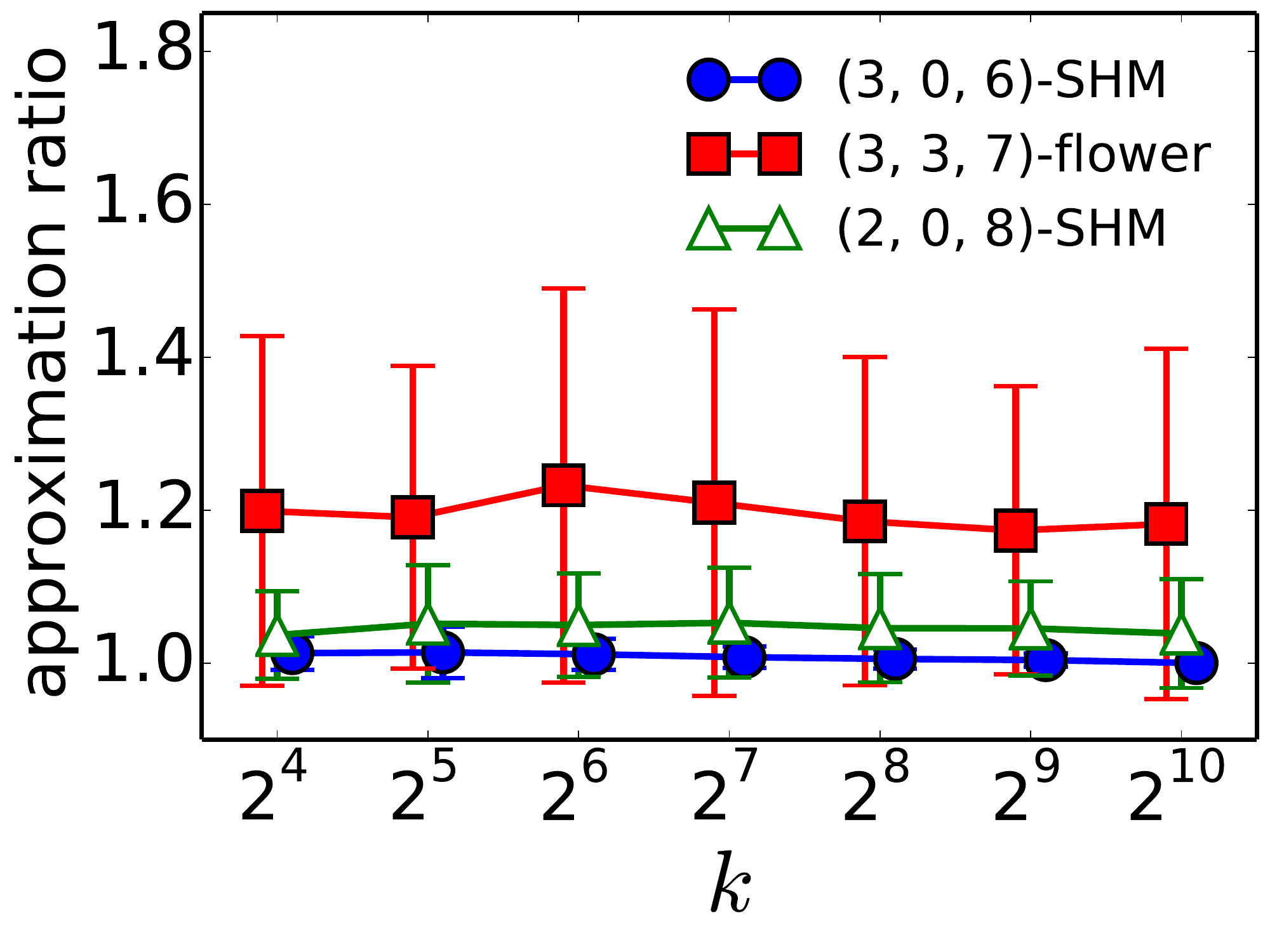}
\includegraphics[width=0.48\hsize]{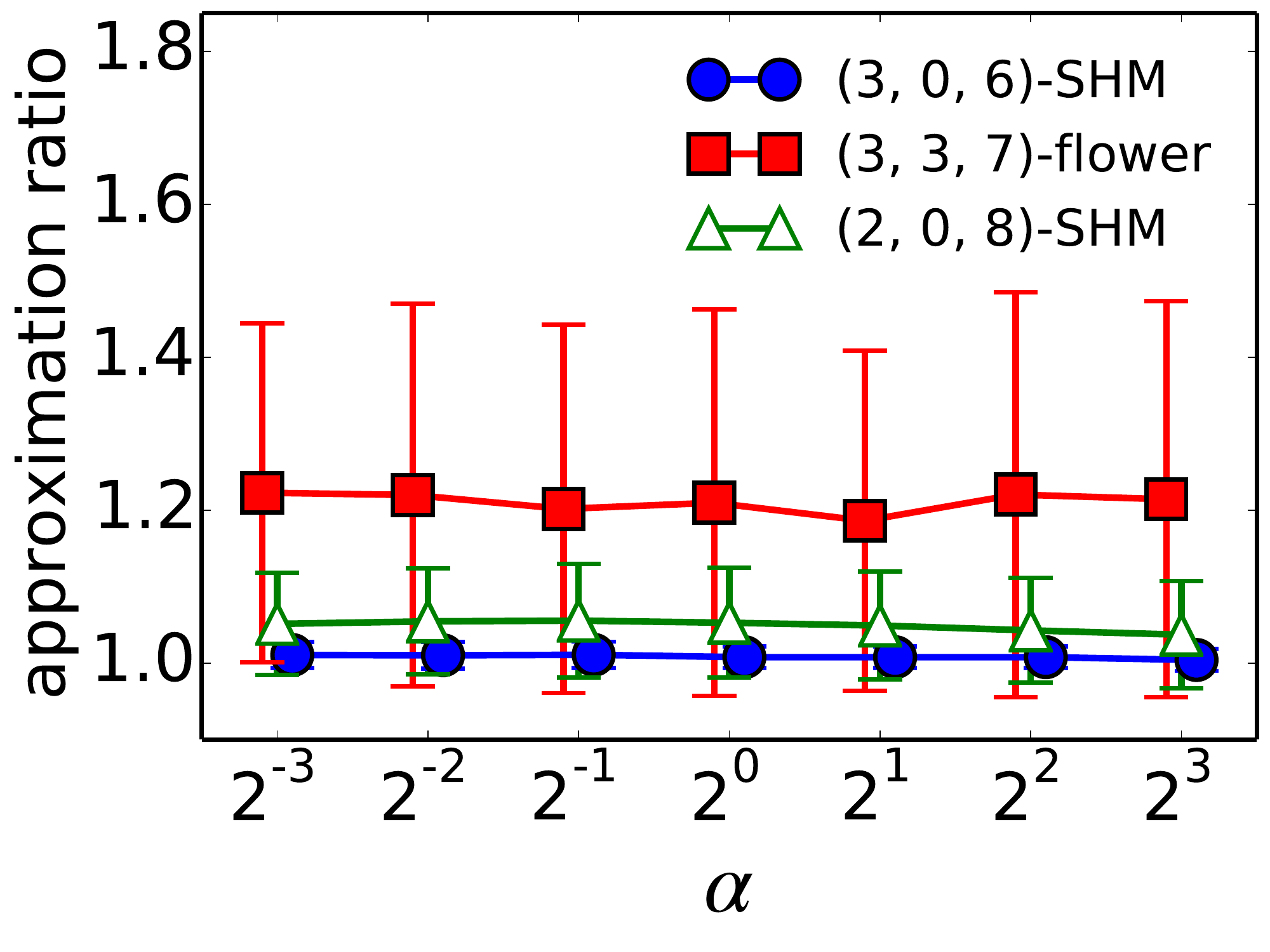}
\vspace{-1.4em}
\caption{Average approximation ratio to the theoretical solutions for various \boldmath{$k$} and \boldmath{$\alpha$}.}
\vspace{-1em}
\label{fig:k_alpha}
\end{figure}

\subsection{Parameter Settings}
First of all, we have to decide the parameter values of our algorithm:
$\epsilon$ (error tolerance),
$k$ (sketch size),
and $\alpha$ (exact neighborhood switch threshold).
In principle, the accuracy of results as well as running time increases with $k$ and $\alpha$, and it decreases with $\epsilon$.
As we discussed in Section~\ref{sec:improve/box_cover},
we fixed $\epsilon = 0$.
To choose $k$ and $\alpha$, we plotted the average approximation ratio of our results to the theoretical solutions for several fractal network models as a function of $k$ and $\alpha$ in Figure~\ref{fig:k_alpha}.
The average approximation ratio is defined by $\rho \equiv \langle b_{\rm sketch}(\ell) / b_{\rm theory}(\ell) \rangle_\ell$, where $\langle \cdot \rangle_\ell$ is the average over $\ell$. To compute $b_{\rm sketch}(\ell)$, we executed the algorithm for ten times and took the average of the resulting $b(\ell)$ over the ten runs.

In the left panel of Figure~\ref{fig:k_alpha}, we varied $2^4 \leq k \leq 2^{10}$ while fixing $\alpha = 1$.
The $\rho$ values were affected slightly by $k$ for the SHM models and tended to decrease with $k$ for the flower network.
On the basis of the results, we decided to use $k=2^7$ throughout the following experiments.
In the right panel of Figure~\ref{fig:k_alpha}, we varied $2^{-3} \leq \alpha \leq 2^{3}$ while fixing $k=2^7$.
The $\rho$ values were almost constant regardless of the $\alpha$ values for all of the three networks considered. Therefore, taking into account the running time, we decided to use $\alpha = 1$ throughout the following experiments.
It is worth noting that Figure~\ref{fig:k_alpha} also demonstrates the high accuracy and robustness of our algorithm over a broad range of parameter values.

\subsection{Accuracy and Scalability}
\label{sec:performance}
Table~\ref{tbl:main_table} summarizes the main results of this paper and shows the comparison of our algorithm ({\emph{Sketch}) with other preceding algorithms for fractal and non-fractal network models with various sizes.
We evaluated the performance of algorithms by two measures.
The first was the accuracy given by $-\log_{10} r_{\rm fit}$ (Section~\ref{sec:procedure}).
If this measure took a positive (negative) value for a fractal (non-fractal) network, the algorithm correctly distinguished the fractality of the network.
The second was computation time in seconds.

\myparagraph{Discrimination Ability}
As we can see in Table~\ref{tbl:main_table}, the sketch algorithm perfectly distinguishes between the fractal and non-fractal networks as the other algorithms do (except for CBB for $(3,4,5)$-flower).
The proposed algorithm shows its advantage in computation time: the algorithm is generally faster than other algorithms and is able to handle large networks that other algorithms do not terminate.
Although MEMB is faster than Sketch for some relatively small network models, this result is expected because actual neighborhood sets are not significantly larger than sketch sizes in these networks.
As a summary, \textit{(i)} the sketch algorithm correctly detected the fractality of network models with around ten times smaller computation time than the fastest previous algorithm.
In addition, \textit{(ii)} the algorithm was able to deal with networks with millions of nodes with acceptable computation time (within $1$ day), whereas other algorithms could not in our machine environment.

\begin{figure}
\centering
\includegraphics[width=0.48\hsize]{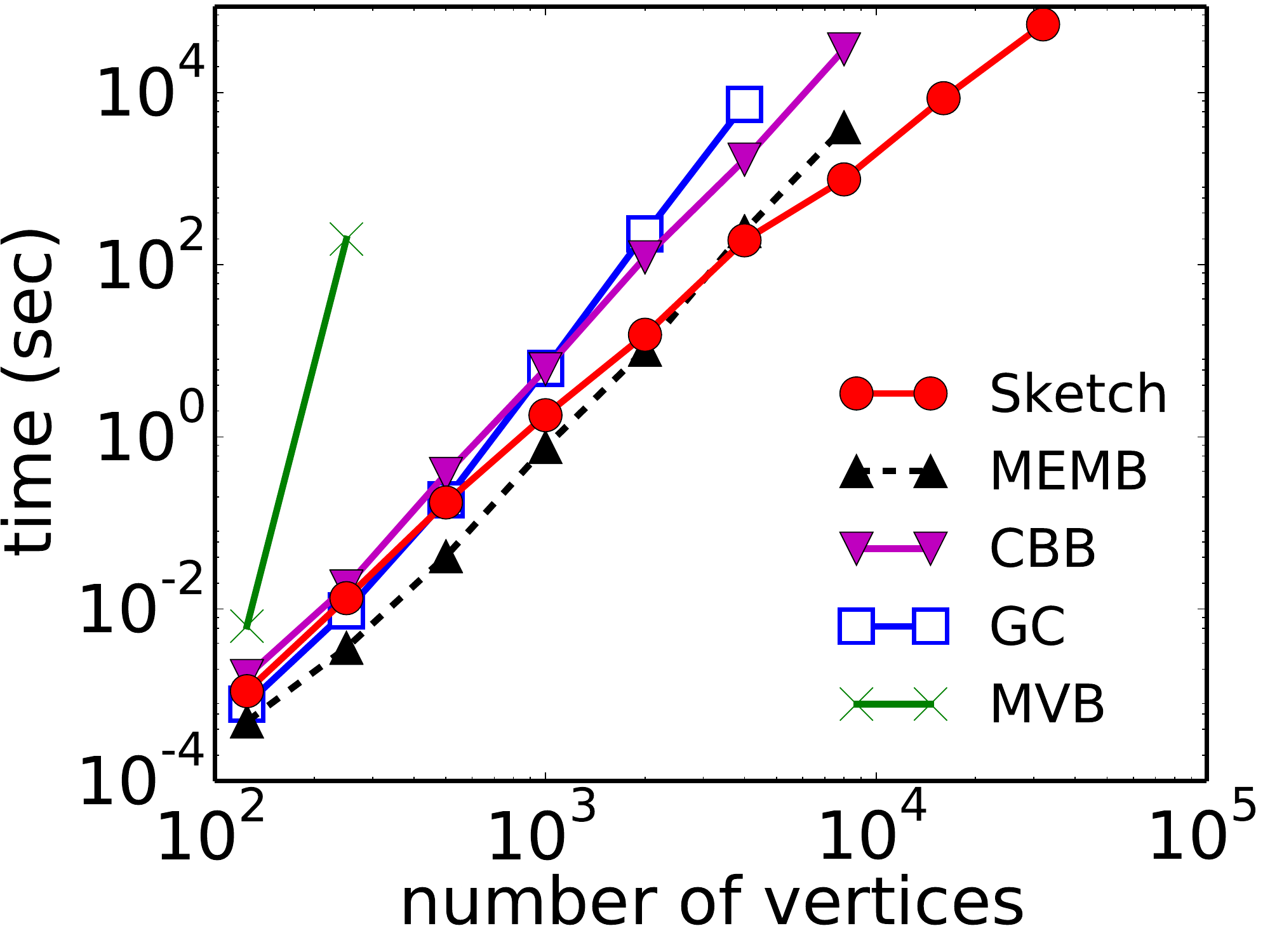}
\includegraphics[width=0.48\hsize]{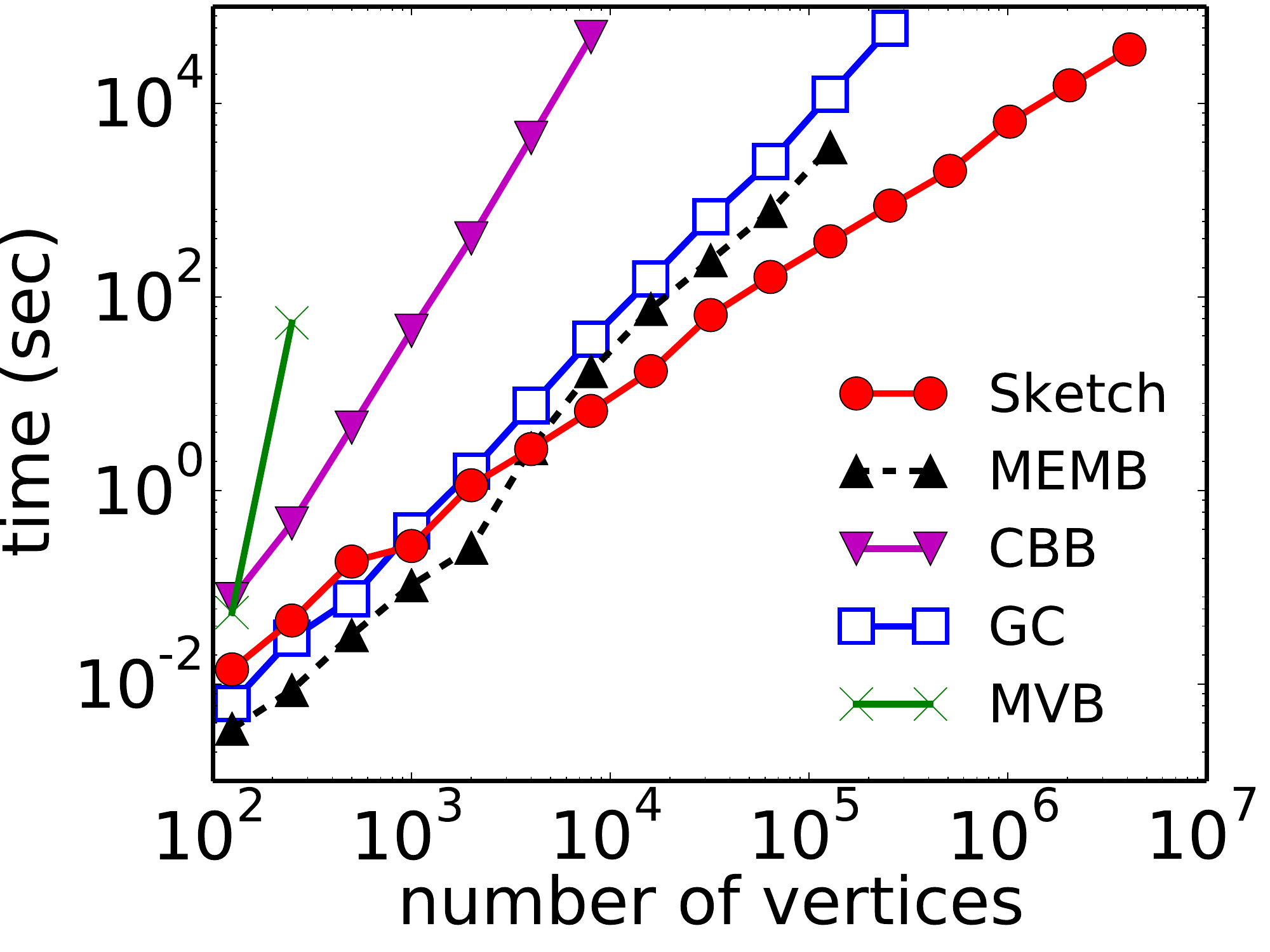}\\
\includegraphics[width=0.48\hsize]{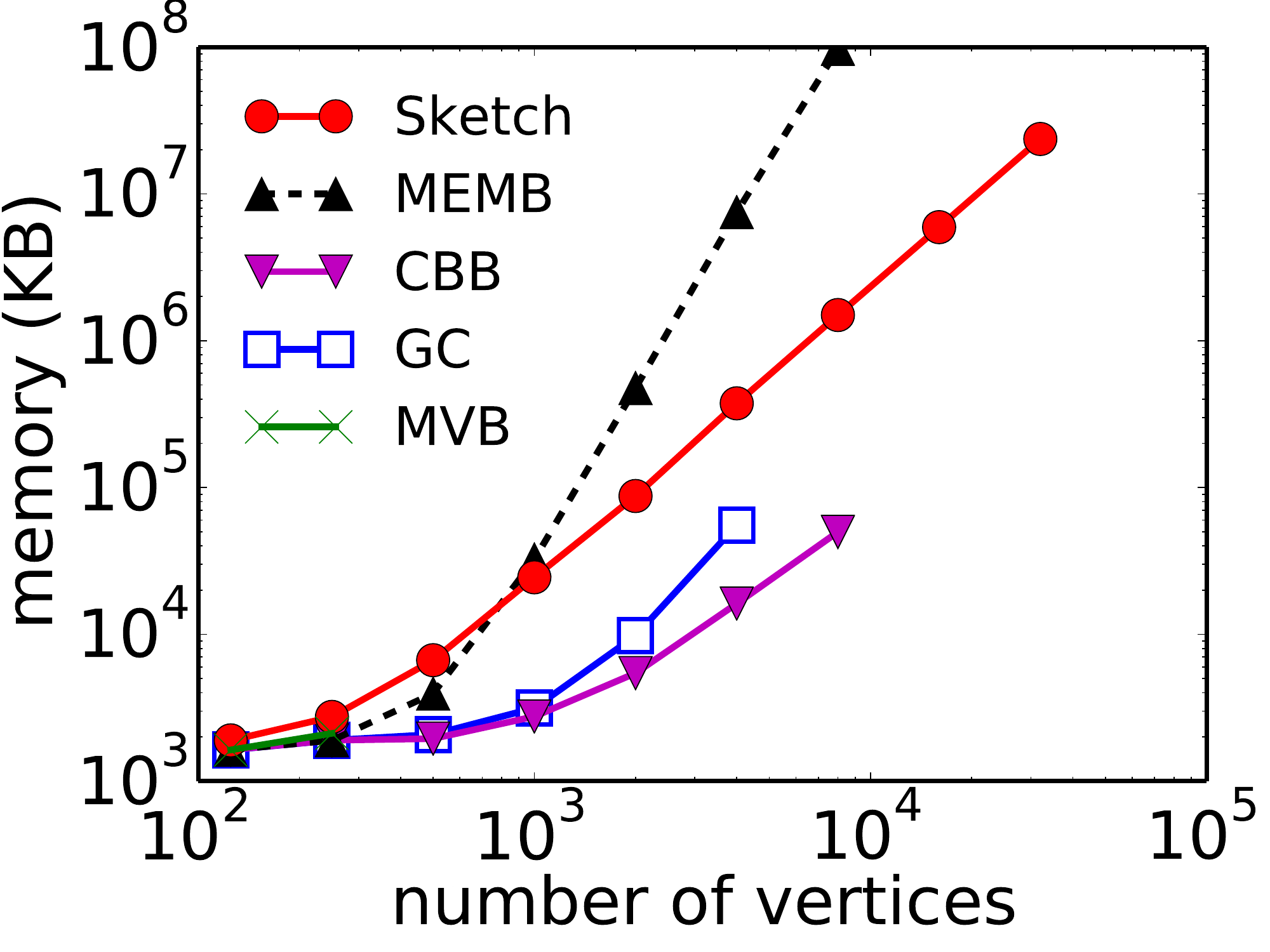}
\includegraphics[width=0.48\hsize]{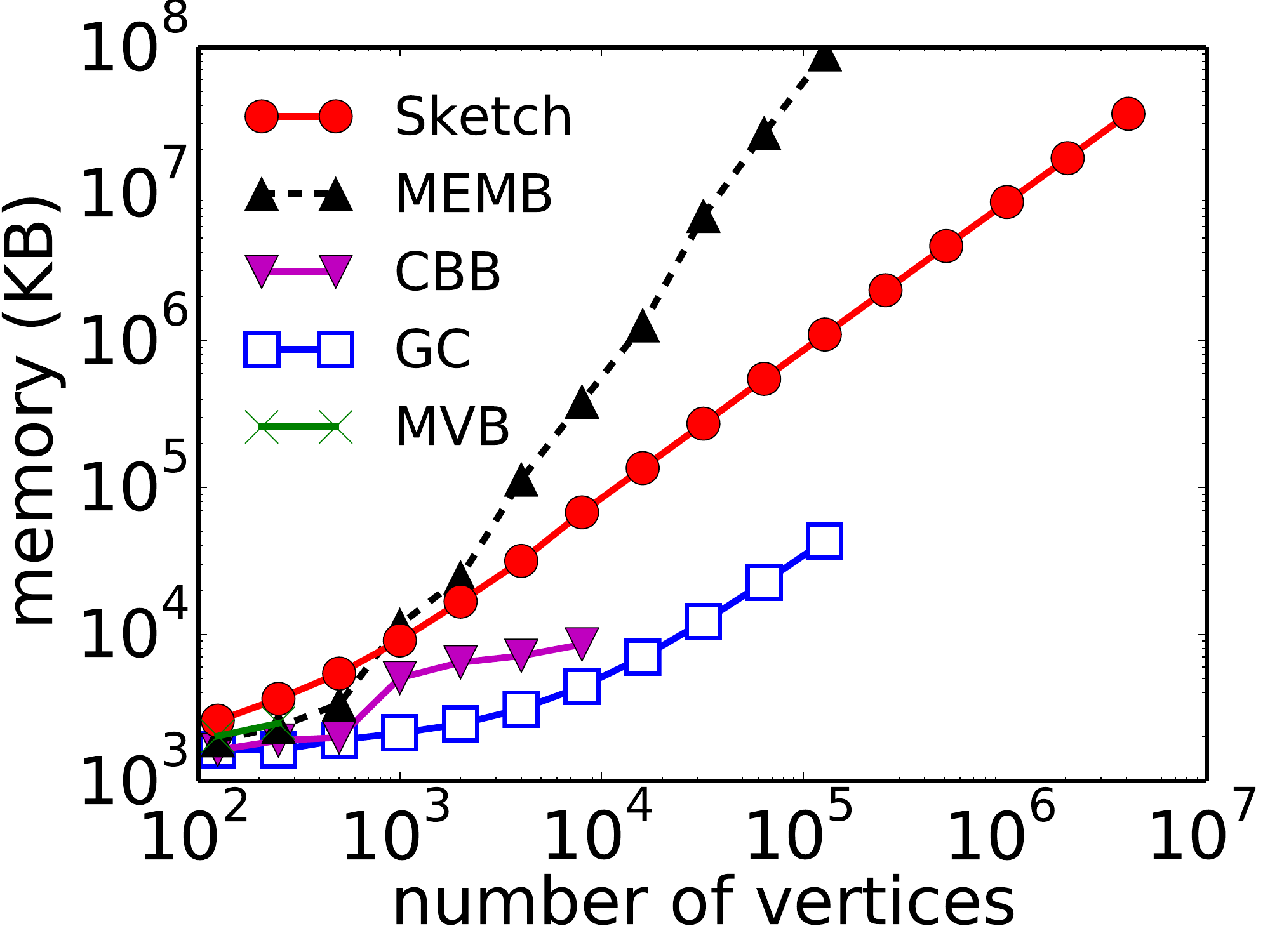}
\vspace{-1em}
\caption{Scalability of computation time (top) and memory usage (bottom) for $(2, 2, g)$-flower (left) and $(2, t)$-BA (right) networks.}
\vspace{-1em}
\label{fig:scalability}
\end{figure}

\myparagraph{Time and Memory Consumption}
The proposed algorithm is scalable for not only for computation time but also for memory usage.
In Figure~\ref{fig:scalability}, computation time (seconds) and memory usage (KB) of the five algorithms were plotted as a function of the number of vertices.
We use $(2,2,g)$-flower $(3 \leq g \leq 11)$ and $(2, t)$-BA $(0 \leq t \leq 15)$ networks as the example of a fractal and a non-fractal network, respectively.
The symbols corresponding to an algorithm were not shown if the algorithm did not stop within 24 hours or could not execute owing to memory shortage.
The performance of the proposed algorithm is comparable to or worse than some other algorithms when the network is relatively small (i.e., $n < 10^4$).
However, the algorithm is orders of magnitude faster than other algorithms for large networks.
Also, it achieves such high a high speed with incomparably smaller memory usage than MEMB, the second fastest algorithm.

\begin{figure}[t!]
\centering
\includegraphics[width=0.48\hsize]{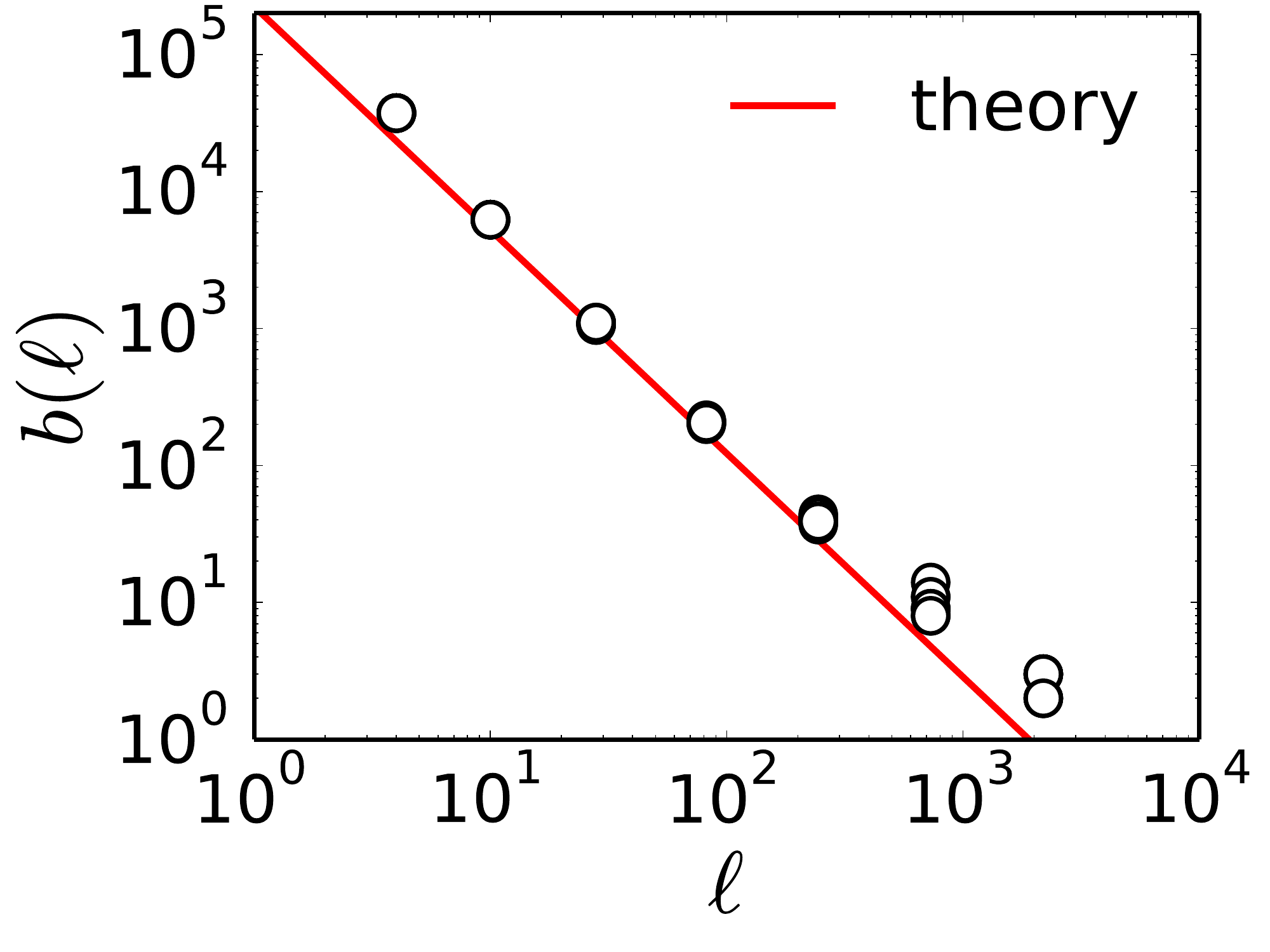}
\includegraphics[width=0.48\hsize]{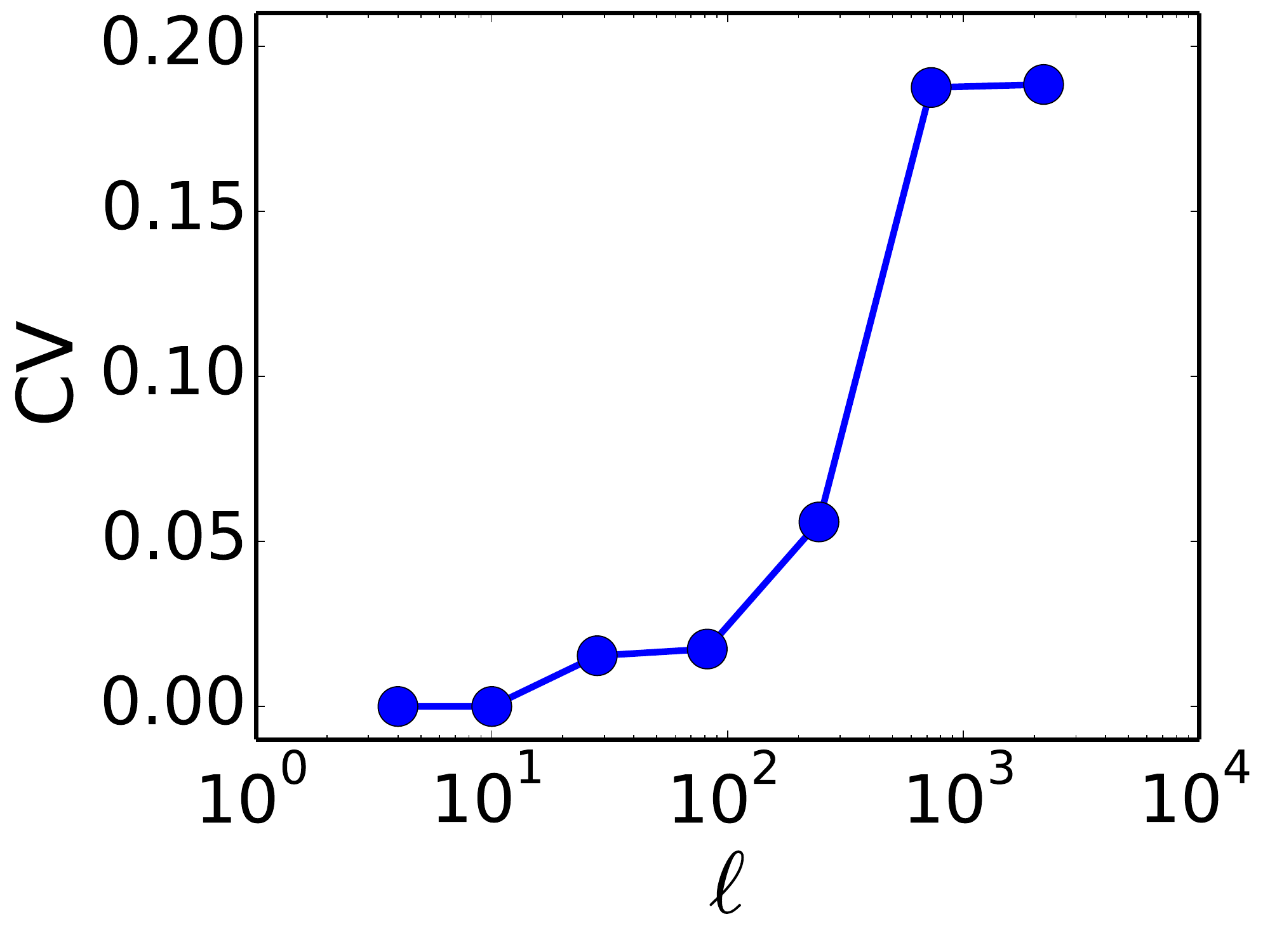}
\vspace{-1em}
\caption{Results of different runs for $(3,3,7)$-flower. (Left) $b(\ell)$ and (Right) CV of them as a function of $\ell$.}
\label{fig:fluctuation}
\vspace{-1em}
\end{figure}

\myparagraph{Robustness over Randomness}
The sketch algorithm accurately recovers $b(\ell)$ of theoretical prediction for fractal network models, and the results are robust over different execution runs.
The left panel of Figure~\ref{fig:fluctuation} shows $b(\ell)$ of ten different runs of the proposed algorithm on $(3,3,7)$-flower.
The $b(\ell)$ values follow well the theoretical solution, which is indicated by the solid line.
As we can clearly observe, the fluctuation in the $b(\ell)$ values due to the randomness is very small.
The consistency over different runs is captured by the CV of $b(\ell)$ (i.e., the ratio of the standard deviation of $b(\ell)$ to its average over ten runs) as a function of $\ell$ (right panel of Figure~\ref{fig:fluctuation}).
The CV values tend to increase with $\ell$. This tendency can be explained by the following two factors.
First, the $b(\ell)$ value takes a positive integer value and monotonically decreases with $\ell$ by definition. Thus, even a change of $\pm 1$ in $b(\ell)$ might cause a large CV value if $\ell$ is large.
Second, our algorithm intrinsically fluctuates more when $\ell$ is larger. This could be because the sizes of the solutions become smaller for larger $\ell$, and hence the algorithm gets a little more sensitive to estimation errors.
Nevertheless, it should be noted that the variance of our algorithm was considerably small even for large $\ell$ (i.e., ${\rm CV} \sim 0.19$ at most). This magnitude of variance would have little impact on the estimation of fractality.

\subsection{Application to Real Large Network}
\label{sec:applications}

In closing this section, we applied the sketch algorithm to a large-scale real graph to show the scalability of the proposed algorithm with an empirical instance.
The results also gave us some insight on the fractality of large-scale real-world networks, which is beyond the reach of previous algorithms.
As a representative instance of a real-world large graph, we considered the
\texttt{in-2004} network~\cite{Boldi2004,Boldi2011}, which is a crawled web graph of $1,382,908$ vertices and $16,917,053$ edges. We discarded the direction of the edges (i.e., hyperlinks) to make the network undirected.
The algorithm took 11.7 hours in total.

The resulting $b(\ell)$ of the sketch algorithm and the fitting curves are shown in Figure~\ref{fig:in-2004}.
We omitted the three points with the smallest $\ell$ values from the fitting because empirical networks would not show a perfect fractality, contrary to well-designed network models.
A large part of the points fall on the line of the fitted power-law function,
and indeed,
our fractality decision procedure yielded $-\log_{10} r_\text{fit} = 0.79$, which suggests the fractality of the \texttt{in-2004} network.
It is worth mentioning that the fractality of this network was unveiled for the first time for the sake of our algorithm.

\begin{figure}
\centering
\includegraphics[width=0.7\hsize]{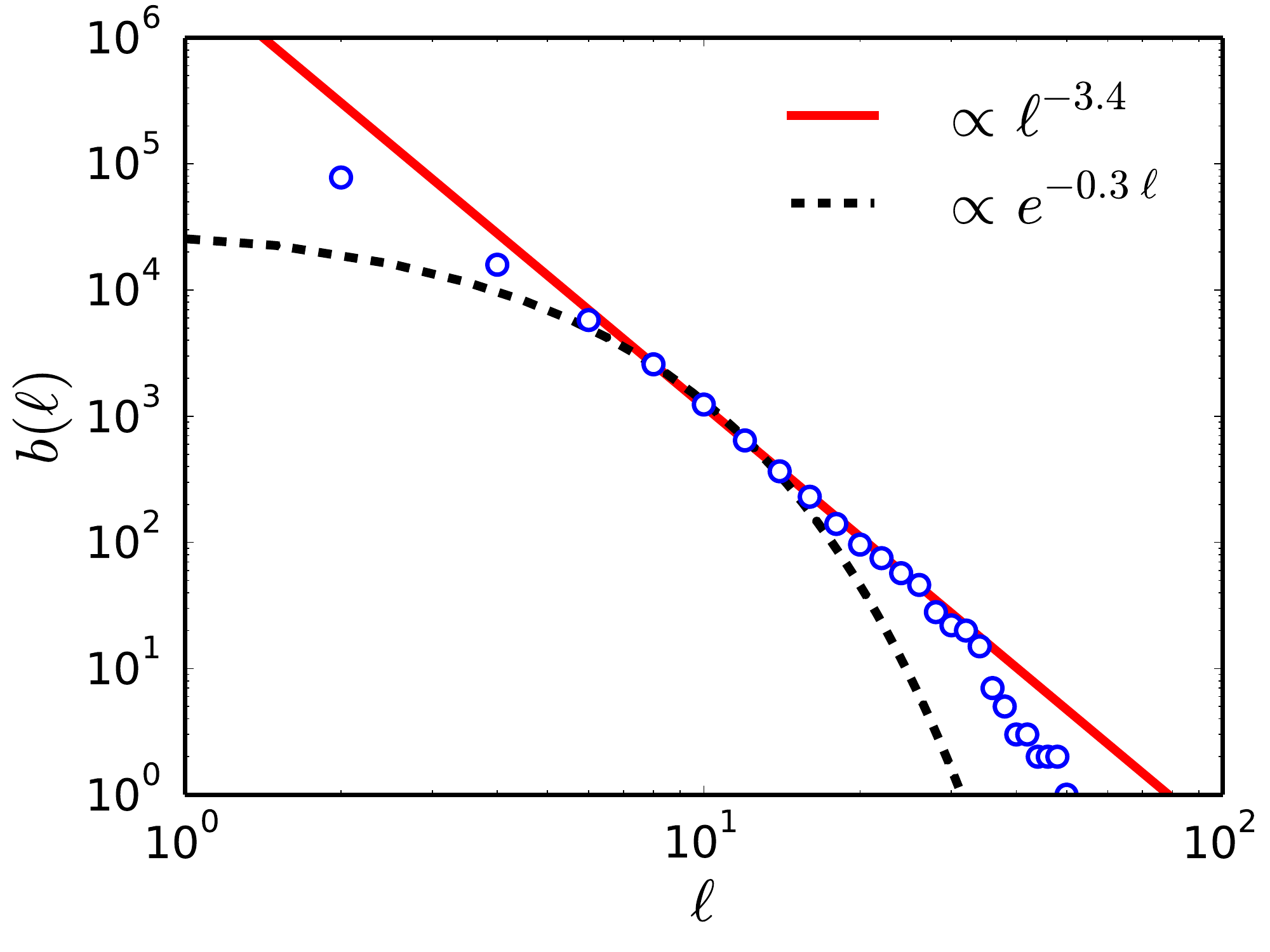}
\vspace{-1.2em}
\caption{Results for a real web graph.}
\vspace{-1em}
\label{fig:in-2004}
\end{figure}

\section{Conclusions}
\label{sec:conclusions}

Fractality is an interesting property that appears in some classes of real networks.
In the present study, we designed a new box-covering algorithm,
which is useful for analyzing the fractality of large-scale networks.
In theory, we have shown desirable guarantees on scalability and solution quality.
In the experiments, we confirmed
that the algorithm's outputs are sufficiently accurate
and that it can handle large networks with millions of vertices and edges.
We hope that our method enables further exploration of graph fractality and its applications
such as graph coarsening.

\myparagraph{Repeatability}
Our implementation of the proposed and previous box-cover algorithms is available at {http://git.io/fractality}.
It also contains the generators of the synthetic network models,
and thus the results in this paper can be perfectly replicated.
 We hope that our public software will enable further exploration of graph fractality and its applications.

\myparagraph{Acknowledgment}
This work was supported by JSPS KAKENHI (No. 15H06828), JST, ERATO, Kawarabayashi Large Graph Project, and JST, PRESTO.
Web graph data was downloaded from \url{http://law.di.unimi.it/datasets.php}.
T.T. thanks to K. Takemoto for valuable discussions.

{
\small
\bibliographystyle{abbrv}
\bibliography{main}
}

\end{document}